\PassOptionsToPackage{svgnames}{xcolor}
\documentclass[11pt,letterpaper]{article}
\usepackage{fullpage}
\usepackage{natbib}
\usepackage{mathpazo}
\usepackage{graphicx}
\usepackage{algorithm}
\usepackage[noend]{algorithmic}
\usepackage{amsmath,amsthm,amsfonts,bbm}
\usepackage[hidelinks,colorlinks,allcolors=DarkBlue,breaklinks]{hyperref}
\usepackage[hyphenbreaks]{breakurl}
\usepackage[nameinlink]{cleveref}
\usepackage{url}
\usepackage{xfrac}
\usepackage{booktabs,multirow}
\usepackage{amsthm,amsmath,amssymb,amsfonts}
\usepackage{mathtools,thmtools}
\usepackage{thm-restate}
\usepackage{cleveref}
\usepackage[dvipsnames]{xcolor}
\usepackage{comment}
\usepackage{todonotes}
\usepackage{tikz}
\usetikzlibrary{arrows.meta}
\usetikzlibrary{positioning}
\usetikzlibrary{calc}
\usepackage{xspace}
\usepackage{nicefrac}
\usepackage[shortlabels]{enumitem}
\usepackage{pdfpages}

\newtheorem{theorem}{Theorem}
\newtheorem{corollary}{Corollary}

\newtheorem{proposition}{Proposition}
\newtheorem{lemma}{Lemma}
\theoremstyle{definition}
\newtheorem{definition}{Definition}
\newtheorem{example}{Example}

\newtheorem{observation}{Observation}

\newtheorem{open}{Open Question}
\usepackage{tcolorbox}
\renewenvironment{open}{\begin{tcolorbox}[colback=gray!7, colframe=black, rounded corners, left=2pt, right=2pt] \vspace*{-2pt} \textbf{Open Question:}}{ \vspace*{-2pt}\end{tcolorbox}}

\makeatletter
\newcommand{\leqnomode}{\tagsleft@true}
\newcommand{\reqnomode}{\tagsleft@false}
\makeatother

\renewcommand{\paragraph}[1]{\smallskip\noindent\textbf{#1}}

\renewcommand{\le}{\leqslant}
\renewcommand{\leq}{\leqslant}
\renewcommand{\ge}{\geqslant}
\renewcommand{\geq}{\geqslant}

\newcommand{\argmin}{\arg\min}

\newcommand{\bbR}{\mathbb{R}}
\newcommand{\bbN}{\mathbb{N}}
\renewcommand{\tilde}{\widetilde}

\DeclarePairedDelimiter\floor{\lfloor}{\rfloor}
\DeclarePairedDelimiter\ceil{\lceil}{\rceil}
\DeclarePairedDelimiter\set{\{}{\}}           
\DeclarePairedDelimiter\abs{\lvert}{\rvert}   

\newcommand{\spliddit}{\url{Spliddit.org}}
\renewcommand{\vec}[1]{#1}
\newcommand{\vals}{\vec{v}}

\newcommand{\NW}{\operatorname{NW}}
\newcommand{\tv}{\tilde{v}}
\newcommand{\allocs}{\mathcal{F}}
\newcommand{\calI}{\mathcal{I}} 
\newcommand{\calC}{\mathcal{C}} 
\newcommand{\half}{\sfrac{1}{2}}
\newcommand{\efwc}{EF1\textsubscript{WC}\xspace}
\newcommand{\poplus}{PO\textsuperscript{+}\xspace}
\newcommand{\card}[1]{\left|{#1}\right|}
\newcommand{\ubar}[1]{\text{\b{$#1$}}}

\newcommand{\efoo}{EF$^1_1$\xspace}
\newcommand{\xx}{x}
\newcommand{\yy}{y}
\newcommand{\pp}{p}

\newcommand{\pref}{\ensuremath{\sigma}}
\newcommand{\hierarchy}{\ensuremath{\mathcal{H}}}

\newcommand{\bb}{\ensuremath{\mathsf{bb}}}
\newcommand{\afull}{\ensuremath{\mathrm{full}}}

\newenvironment{myquote}%
  {\list{}{\leftmargin=0.3in\rightmargin=0.3in}\item[]}%
  {\endlist}

\title{Constrained Fair and Efficient Allocations}
\author{Benjamin Cookson \and Soroush Ebadian \and Nisarg Shah}
\date{%
    University of Toronto\\%
    \{bcookson,soroush,nisarg\}@cs.toronto.edu
}

\begin{document}
\maketitle

\begin{abstract}
    Fairness and efficiency have become the pillars of modern fair division research, but prior work on achieving both simultaneously is largely limited to the \emph{unconstrained} setting. We study fair and efficient allocations of indivisible goods under additive valuations and various types of allocation feasibility constraints, and demonstrate the unreasonable effectiveness of the maximum Nash welfare (MNW) solution in this previously uncharted territory. 
    
    Our main result is that MNW allocations are \half-envy-free up to one good (EF1) and Pareto optimal under the broad family of (arbitrary) matroid constraints. We extend these guarantees to \emph{complete} MNW allocations for base-orderable matroid constraints, and to a family of non-matroidal constraints (which includes balancedness) using a novel ``alternate worlds'' technique. We establish tightness of our results by providing counterexamples for the satisfiability of certain stronger desiderata, but show an improved result for the special case of goods with copies~\citep{GHLT24}. Finally, we also establish novel best-of-both-worlds guarantees for goods with copies and balancedness. 
\end{abstract}

\section{Introduction}\label{sec:introduction}
Fair division of resources among agents is a primitive that has applications, both to multiagent systems~\citep{CDEL+06} and to everyday problems such as estate division and divorce settlement~\citep{Shah17}. Over the last decade, the fair division literature has undergone a dramatic transformation. The pioneering work of \cite{CKMP+19} established that, under additive valuations, the so-called maximum Nash welfare (MNW) allocations, which (informally) maximize the product of agent utilities, simultaneously satisfy two appealing guarantees: a fairness criterion known as \emph{envy-freeness up to one good} (EF1), which demands that no agent prefer the allocation of another agent (modulo a single good) to her own, and an efficiency criterion known as \emph{Pareto optimality} (PO), which demands that no alternative allocation be able to make an agent happier without making any agent worse off. These provable fairness and efficiency guarantees have been critical to their use in the real world via the not-for-profit website \spliddit~\citep{Shah17}. 

Ever since then, the combination of fairness and efficiency, in the form of approximate envy-freeness and Pareto optimality, has become the guiding principle for seeking fair division solutions, e.g., for subclasses of additive valuations~\citep{HSVX21}, for non-additive valuations~\citep{BCIZ21,BS24}, when addressing manipulations~\citep{PV22}, or for allocating chores~\citep{EPS22,GMQ22}, or for allocating public goods~\citep{FMS18,EFS24}. 

However, in many real-world fair division problems, there are feasibility constraints on the bundle that each agent can receive. Examples include course allocation~\citep{BCKO17}, public housing assignment~\citep{BCHS+20}, or allocation of conference submissions to reviewers~\citep{garg2010assigning}. Unfortunately, the literature on \emph{constrained} fair division has been largely limited to seeking only fairness guarantees. 
\begin{itemize}
\item \cite{biswas2018fair} show the existence of an EF1 allocation subject to cardinality constraints, where the goods are partitioned into categories and each agent must be allocated at most a prescribed maximum number of goods from each category.
\item \cite{biswas2019matroid} extend this to any base-orderable matroid constraint, where the bundle of goods allocated to each agent must be an independent set of a given base-orderable matroid (cardinality constraints form a partition matroid, which is a special case), when agents have identical additive valuations.\footnote{\cite{biswas2019matroid} incorrectly state their result for an arbitrary matroid constraint in the original paper, but later versions correctly state that the result holds for base-orderable matroids. The existence of an EF1 allocation here remains open for general matroid constraints, even for identical additive valuations.}
\item \cite{GHLT24} study a special case of cardinality constraints, which they refer to as \emph{goods with copies}, motivated by the fact that it in turn subsumes chore division as a special case. They define an appealing strengthening of EF1, termed \efwc, and establish its existence for restricted valuation classes.  
\item The popular round robin algorithm yields an EF1 allocation subject to \emph{balancedness}, where all agents must be assigned bundles of roughly equal cardinality (differing by at most one)~\citep{CKMP+19}. 
\item A famous non-matroidal constraint is where the goods are vertices of an undirected graph, and the bundle allocated to each agent must form a connected subset; when the graph is a line, an EF1 allocation is known to exist~\citep{igarashi2023cut}, and for general graphs, a \half-EF1 allocation with up to $n-1$ unallocated goods is known to exist under restricted preferences~\citep{CMS22}. 
\end{itemize}
In addition to the above summary of related work, we provide comparisons to other pieces of related work throughout the paper, and also refer the reader to the extensive survey on constraints in fair division by \cite{suksompong2021constraints}. Quite surprisingly, the existence of an allocation that satisfies both (even approximate) EF1 and PO remains severely understudied in these constrained domains. The only exception is the work on budget constraints, where each good has a size and the total size of goods allocated to any agent must be at most a threshold. \cite{gan2023approximation} show that an EF1 allocation always exists, and \cite{wu2021budget} show that any MNW allocation subject to such a constraint is \sfrac{1}{4}-EF1 and PO. Our main research question is to expand on this line of work:
\begin{myquote}
\emph{Under which types of feasibility constraints do (approximately) fair and efficient allocations exist?}	
\end{myquote}

\subsection{Our Results}\label{sec:results}
Our main result is that every MNW allocation is \half-EF1 and PO under the broad class of (arbitrary) matroid constraints (see \Cref{sec:prelims} for a formal definition). While these allocations are efficient, they can be \emph{incomplete}, i.e., leave some goods unallocated, which may be undesirable in settings such as allocation of shifts to nurses and assignment of conference submissions to reviewers. To that end, we show that for base-orderable matroid constraints, even allocations that are MNW among the set of \emph{complete} and feasible allocations are \half-EF1 and PO. Base-orderable matroids subsume the case of cardinality constraints~\citep{biswas2018fair}.  

Then, using a novel technique of constructing ``alternate worlds'', we show that the \half-EF1 and PO guarantees can be extended to MNW allocations subject to a broad class of non-matroidal constraints, which includes balancedness as a special case. We also show that certain strengthenings of EF1 and PO are unachievable in the realm of constrained allocations, but show an improvement from EF1 to \efwc for the case of goods with copies~\citep{GHLT24}. 

Finally, we expand the recent work on ``best of both worlds'' (BoBW) guarantees~\citep{AFSV24} to the realm of constrained allocations. Building on the work of \cite{echenique2021constrained}, we prove that randomized allocations that are ex ante EF and PO along with ex post \efoo, Prop1, and PO exist for the case of goods with copies, and the same result except for the \efoo guarantee holds when adding the balancedness constraint. For formal definitions, see \Cref{sec:bobw}.

\section{Preliminaries}\label{sec:prelims}
For any $r \in \bbN$, define $[r] := \set{1,2,\dots,r}$. Let $N = [n]$ be a set of \emph{agents}, and $M$ be a set of $m$ \emph{(indivisible) goods}. Each agent $i$ has a valuation function $v_i : M \to \bbR_{\ge 0}$, where $v_i(g)$ is her value for good $g$. We assume additive valuations: with slight abuse of notation, the value of agent $i$ for a set of goods $S \subseteq M$ is $v_i(S) := \sum_{g \in S} v_i(g)$. Let $\vals = (v_1,\ldots,v_n)$ be the \emph{valuation profile}. In an \emph{allocation} $A = (A_1,\ldots,A_n)$, $A_i \subseteq M$ is the bundle of goods assigned to agent $i$ and $A_i \cap A_j = \emptyset$ for all distinct $i,j \in M$. We say that $A$ is \emph{complete} if $\cup_{i \in N} A_i = M$. The utility to agent $i$ under this allocation is $v_i(A_i)$. 

\subsection{Feasibility Constraints}
We are interested in \emph{constrained} allocation problems, where we are allowed to choose an allocation only from a given set of allocations $\allocs$. One of the most general family of constraints uses matroids. 

\paragraph{Matroid constraints.} We are given a matroid over the goods $(M,\calI)$ satisfying: (a) $\emptyset \in \calI$; (b) \emph{(hereditary property)} if $S \in \calI$ and $S' \subseteq S$, then $S' \in \calI$; and (c) \emph{(exchange property)} if $S,T \in \calI$ and $|T| > |S|$, then there exists $g \in T \setminus S$ such that $S \cup \set{g} \in \calI$. Then, an allocation $A$ is called feasible (i.e., $A \in \allocs$) if $A_i \in \calI$ for all agents $i \in N$. Bases of a matroid are its maximal independent sets. The exchange property implies that all bases have the same cardinality (known as the rank of the matroid), and for any two bases $S$ and $T$, there exist $g \in S \setminus T$ and $g' \in T \setminus S$ such that $S \cup \set{g'} \setminus \set{g}$ and $T \cup \set{g} \setminus \set{g'}$ are both bases as well. We will reference the following popular families of matroids in this paper. 

\begin{itemize}
    \item \emph{Base-orderable matroids.} These are matroids with a strengthened base exchange property: for any two bases $S$ and $T$, there exists a \emph{bijection} $f: S \to T$ such that, for every $g \in S$, $S \cup \set{f(g)} \setminus \set{g}$ and $T \cup \set{g} \setminus \set{f(g)}$ are both bases as well.
    \item \emph{Laminar matroids.} This is a special case of base-orderable matroids, where we are given a laminar family of subsets of goods (termed \emph{categories}) $\calC \subseteq 2^M$. Here, laminar means that for every $C,C' \in \calC$, we have $C \cap C' \in \set{\emptyset,C,C'}$. For each category $C \in \calC$, we are also given an upper bound $h_C$. Then, $\calI$ consists of all sets $S$ such that $|S \cap C| \le h_C$ for every category $C \in \calC$. For feasibility, we assume $h_C \ge \ceil{\nicefrac{|C|}{n}}$ for each $C \in \calC$.
    \item \emph{Partition matroids.} This is a special case of laminar matroids, where $\calC$ is a partition of the set of goods $M$, i.e., $C \cap C' = \emptyset$ for all $C,C' \in \calC$ and $\cup_{C \in \calC} C = M$. This is also known in the fair division literature as \emph{cardinality constraints}~\citep{biswas2018fair}. 
    \item \emph{Goods with copies.} Introduced by \cite{GHLT24}, this is a special case of partition matroids, where $h_C = 1$ and $|C| \le n$ for each category $C \in \calC$. The typical motivation for this is when the goods in each category are copies (i.e., perfect substitutes) of each other, i.e., $v_i(g) = v_i(g')$ for all $i \in N$, $C \in \calC$, and $g,g' \in C$. This is what we will assume when referring to the case of goods with copies.

\end{itemize}

\paragraph{Non-matroidal constraints.} We also study the following constraints that do not form a matroid.
\begin{itemize}
    \item \emph{Partition matroids with lower bounds.} Here, we have a partition matroid induced by the partition $\calC$ of $M$ and upper bounds $(h_C)_{C \in \calC}$. Additionally, we also have lower bounds $(\ell_C)_{C \in \calC}$. Then, $\allocs$ consists of the set of allocations $A$ where $\ell_C \le |A_i \cap C| \le h_C$ for every $i \in N$ and $C \in \calC$. For feasibility, we assume $\ell_C \le \floor{\nicefrac{|C|}{n}} \le \ceil{\nicefrac{|C|}{n}} \le h_C$ for each $C \in \calC$. 
    \item \emph{Balanced.} This is a special case of partition matroids with lower bounds with a singleton $\calC = \set{M}$, $\ell_M = \floor{\sfrac{m}{n}}$ and $h_M = \ceil{\sfrac{m}{n}}$. That is, balanced allocations $A$ satisfy $|A_i| \in \set{\floor{\sfrac{m}{n}}, \ceil{\sfrac{m}{n}}}$ for all $i \in N$, so the number of goods allocated to any two agents differ by at most one. 
\end{itemize}

\subsection{Maximum Nash Welfare}
\cite{CKMP+19} introduced the \emph{maximum Nash welfare} rule and proved that, given any unconstrained fair division instance with additive valuations, it always returns an EF1+PO allocation (thus settling the open question of the existence of such allocations). The following is the natural adaptation of their rule to the \emph{constrained} case.

\begin{definition}[(Constrained) Maximum Nash Welfare (MNW)]
For an allocation $A$, define $P(A) = \set{i \in N: v_i(A_i) > 0}$ to be the set of agents receiving a positive utility. An allocation $A \in \allocs$ is a \emph{(constrained) maximum Nash welfare} (MNW) allocation if it maximizes the number of agents receiving positive utility, $\abs{P(A)}$, and, subject to that, maximizes the product of positive agent utilities (Nash welfare), $\NW(A) := \prod_{i \in P(A)} v_i(A_i)$.\footnote{Note that in the typical case where there exists an allocation yielding a positive utility to every agent, this simplifies to saying that $A$ is an MNW allocation if it maximizes $\prod_{i \in N} v_i(A_i)$.}  
\end{definition}

When the above definition is applied to the set of \emph{complete and feasible} allocations, we refer to the resulting allocations as \emph{complete MNW allocations}. Note that a complete MNW allocation may not necessarily be an MNW allocation (among the set of all feasible allocations); see \Cref{ex:mnw-incomplete}. 

\subsection{Fairness and Efficiency Desiderata} 
We are interested in allocations that satisfy various desiderata. The two that play a key role are the following.

\begin{definition}[Approximate envy-freeness up to one good (EF1)]
    For $\alpha \in [0,1]$, an allocation $A \in \allocs$ is $\alpha$-approximate \emph{envy-free up to one good} ($\alpha$-EF1) if, for every pair of agents $i,j \in N$, either $A_j = \emptyset$ or $v_i(A_i) \ge \alpha \cdot v_i(A_j \setminus \set{g})$ for some $g \in A_j$. In words, agent $i$ should not envy agent $j$, up to a factor of $\alpha$, after excluding some good from agent $j$'s bundle. When $\alpha=1$, we simply call it an EF1 allocation. 
\end{definition}

\begin{definition}[Approximate (constrained) Pareto optimality (PO)]
    For $\alpha \in [0,1]$, an allocation $A \in \allocs$ is $\alpha$-approximate \emph{(constrained) Pareto optimal} ($\alpha$-PO) if there is no other allocation $B \in \allocs$ such that $v_i(B_i) \ge \alpha \cdot v_i(A_i)$ for each $i \in N$ and at least one inequality is strict. In words, no other feasible allocation should be able to make an agent happier without making any agent less happy. When $\alpha=1$, we simply call it a PO allocation.  
\end{definition}

We also study relaxations and strengthenings of these two desiderata, but we define them in their respective sections.

\section{The Unreasonable Fairness of (Constrained) Maximum Nash Welfare}\label{sec:nash}

A common technique to deal with feasibility constraints is to reduce the problem to an unconstrained instance by modifying the valuation of each agent for any bundle $S$ to be her highest value for any feasible subset $T \subseteq S$, computing a desirable allocation $A$ for the unconstrained instance, and then giving to each agent $i$ her most valuable feasible subset of $A_i$. \cite{biswas2018fair} observed that such a reduction for partition matroid constraints induces submodular valuations.\footnote{Valuation $v$ is submodular if $v(S \cup T) + v(S \cap T) \le v(S)+v(T)$ for all $S,T \subseteq M$.} However, this actually holds true for any matroid constraint due to known results from matroid theory. In addition to this technique, we use a result by \cite{CKMP+19} that every MNW allocation for an unconstrained instance with submodular valuations satisfies a relaxation of EF1 that they term \emph{marginal envy-freeness up to one good} (MEF1), and observe that MEF1 implies \half-EF1. 

\begin{restatable}{theorem}{matroid}\label{thm:matroid}
Under any matroid constraint, every MNW allocation is \half-EF1 and PO.
\end{restatable}
\begin{proof}
For this, we can reduce the constrained problem with additive valuations to an unconstrained problem with non-additive valuations using a popular trick (see, e.g., the works of \cite{biswas2018fair} and \citet{dror2023fair}). Consider an instance with additive valuations $v = (v_1,\ldots,v_n)$ and matroid constraint $(M,\calI)$. In the new unconstrained instance, the valuation function of agent $i$ is given by $\tv_i(S) = \max_{T \subseteq S : T \in \calI} v_i(T)$; that is, each agent values $S$ by the best feasible bundle within $S$. This construction is known as a \emph{weighted rank function} of a matroid with $v_i(g)$ being the weight of element $g$, and it is known that any weighted rank function of any matroid is monotone submodular~\citep{schrijver2003combinatorial}. 

Next, consider any MNW allocation $A$ in the constrained instance. We consider $A$ as an allocation in the unconstrained instance. Next, we show that $A$ is also an MNW allocation in the unconstrained instance. Suppose this is not true. Then, there exists an MNW allocation $\tilde{B}$ in the unconstrained instance such that one of two things happens:

\begin{enumerate}[(1)]
\item $\begin{aligned}
|P(\tilde{B})| 
= |\set{i \in N : \tv_i(\tilde{B}_i) > 0}|
> |\set{i \in N : \tv_i(A_i) > 0}| = |P(A)|
\end{aligned}$
\item $P(\tilde{B}) = P(A) = P$ and $\prod_{i \in P} \tv_i(\tilde{B}_i) > \prod_{i \in P} \tv_i(A_i)$.
\end{enumerate}

Generate an allocation $B$ in the constrained instance by giving to each agent $i$ her most valuable feasible subset of $\tilde{B}_i$. Notice that $\tv_i(A_i) = v_i(A_i)$ and $\tv_i(\tilde{B}_i) = v_i(B_i)$ for each $i \in N$. Now, consider allocations $A$ and $B$ in the constrained instance. In the first case above, we would get $|\set{i \in N : v_i(B_i) > 0}| > |\set{i \in N : v_i(A_i) > 0}|$. And in the second case above, we would get $\set{i \in N : v_i(B_i) > 0} = \set{i \in N : v_i(A_i) > 0} = P$ (say) and $\prod_{i \in P} v_i(B_i) > \prod_{i \in P} v_i(A_i)$. Both would violate the fact that $A$ is an MNW allocation in the constrained instance.

Finally, \cite{CKMP+19} show that every MNW allocation in an unconstrained instance with submodular valuations $\tv = (\tv_1,\ldots,\tv_n)$, including allocation $A$ above, satisfies marginal envy-freeness up to one good (MEF1): for every pair of agents $i,j \in N$, either $A_j = \emptyset$ or there exists $g \in A_j$ such that $\tv_i(A_i \cup A_j \setminus \set{g}) - \tv_i(A_i) \le \tv_i(A_i)$. For monotone submodular valuations, it is easy to see that MEF1 implies \half-EF1: $\tv_i(A_j \setminus \set{g}) \le 2 \tv_i(A_i)$. Noticing that $\tv_i(A_i) = v_i(A_i)$ and $\tv_i(A_j) = v_i(A_j)$ yields that $A$ is \half-EF1 in the constrained instance too.

The fact that $A$ is PO is established similarly. If there exists an allocation $B$ of the constrained instance that Pareto dominates $A$, then, even in the unconstrained instance, $B$ would Pareto dominate $A$, contradicting the fact that $A$, being an MNW allocation in the unconstrained instance, is PO~\citep{CKMP+19}. 
\end{proof}

An MNW allocation from \Cref{thm:matroid} may be incomplete. At first, one may worry that such an allocation can be highly inefficient. However, note that not even a complete allocation can Pareto dominate an MNW allocation, nor have a higher Nash welfare. As a consequence, every agent $i$ must have zero value for every unallocated good $g$ that can be feasibly added to $A_i$ (otherwise adding $g$ to $A_i$ would yield a Pareto improvement), and no agent $i$ can envy any feasible bundle from the unallocated goods (otherwise swapping such a bundle with $A_i$ would yield a Pareto improvement).

Nonetheless, as motivated in \Cref{sec:introduction}, there are settings in which complete allocations are desirable because free disposal of goods may not be allowed. \emph{Does every matroid constrained instance admit a complete allocation that is \half-EF1 and PO?} The only reason \Cref{thm:matroid} does not already answer this positively is that there may exist an instance in which \emph{every} MNW allocation is incomplete. The following example shows that such instances indeed exist!

\begin{example}\label{ex:mnw-incomplete}
    Consider an instance with $2$ agents, $8$ goods ($g_1$ through $g_8$), and the following laminar matroid constraint:
    \begin{itemize}[left=0pt]
        \item Category $C_1 = \set{g_1,g_2,g_3,g_4}$ has upper bound $h_{C_1} = 2$.
        \item Category $C_2 = \set{g_1,g_2,g_3,g_4,g_5,g_6,g_7,g_8}$ has upper bound $h_{C_2} = 4$.
    \end{itemize}
    The valuations of the agents are as follows.
    \begin{center}
    \renewcommand{\arraystretch}{1.2}
    \begin{tabular}{c | c c c c c c c c}
        \toprule
               & $g_1$ & $g_2$ & $g_3$ & $g_4$ & $g_5$ & $g_6$ & $g_7$ & $g_8$ \\
        \midrule
         Agent $1$ & 0 & 1 & 0 & 0 & 1 & 1 & 1 & 0 \\
         Agent $2$ & 0 & 0 & 1 & 1 & 0 & 0 & 0 & 1 \\
         \bottomrule
    \end{tabular}
    \end{center}

    Note that there is a unique MNW allocation:
    \[A_1 = \set{g_2,g_5,g_6,g_7}, A_2 = \set{g_3,g_4,g_8}.\]
\end{example}

In the above example, notice that every \emph{complete} allocation is Pareto dominated by the incomplete MNW allocation. A similar result is observed by \citet{BCIZ21}, where they show that for binary submodular valuations, some instances have no allocation that maximizes utilitarian welfare that are both complete and ``clean'' (where clean means that no agent's bundle includes a good which that agent has $0$ marginal utility for). Under the paradigm of reducing instances under matroid constraints to submodular valuations, one can easily see that in order for the allocation in the original instance to be complete, the allocation over the reduced submodular instance must be complete and clean (with respect to all goods that do not have $0$ additive value in the original instance). 

This is why completeness should only be sought if free disposal is indeed impossible. Also, due to the above example, in what follows, Pareto optimality of complete allocations would mean only that they are not Pareto dominated by other \emph{complete} and feasible allocations. 

We now show that at least for the broad family of \emph{base-orderable} matroid constraints, \emph{complete MNW allocations} --- recall that these are MNW allocations among the set of complete and feasible allocations --- are also \half-EF1 and PO. 

\begin{theorem}\label{thm:base-orderable}
Under any base-orderable matroid constraints, every complete MNW allocation is \half-EF1 and PO.	
\end{theorem}
\begin{proof}
Consider any instance with additive valuations $v = (v_1,\ldots,v_n)$ and a base-orderable matroid constraint $(M,\calI)$. Let $A$ be any complete MNW allocation. 

To see that $A$ is PO, suppose for contradiction that a complete allocation $B$ Pareto dominates $A$. Then, $B$ must either give a positive utility to strictly more agents ($|P(B)| > |P(A)|$) or give a positive utility to the same set of agents ($P(B) = P(A)$) while increasing their product of utilities ($\NW(B) > \NW(A)$). Both possibilities contradict the fact that $A$ is a complete MNW allocation. 

Next, we show that $A$ is \half-EF1. Suppose for contradiction that there exist agents $i,j \in N$ such that $A_j \neq \emptyset$ and $2\cdot v_i(A_i) < v_i(A_j) - v_i(g)$ for all $g \in A_j$. 

Due to Observation $2$ from \cite{dror2023fair}, we can assume without loss of generality that $A_i$ is a basis of the matroid $(M,\calI)$ for each $i \in N$. This is achieved by performing a preprocessing step of adding dummy goods to the instance. The details of this step along with a proof that complete MNW allocations in the original instance and in the preprocessed instance are equivalent is provided in \Cref{app:preprocessing}.

Now, since $A_i$ and $A_j$ are both bases, let $f: A_j \to A_i$ be the bijection from the definition of base-orderability of $(M,\calI)$. Define $A^*_j := \set{g \in A_j \mid v_i(g) > v_i(f(g))}$ to be the set of all goods $g \in A_j$ such that agent $i$'s utility would increase if $g$ and $f(g)$ were swapped between $A_i$ and $A_j$. Define its projection $A^*_i := \set{f(g) \mid g \in A^*_j}$. Then:
\begin{align}
&\sum\nolimits_{g \in A^*_j}{(v_i(g) - v_i(f(g)))} \nonumber\\
&\quad = v_i(A^*_j) - v_i(A^*_i) \ge v_i(A_j) - v_i(A_i)\nonumber\\
&\quad > v_i(A_i)+\max\nolimits_{g' \in A^*_j} {(v_i(g')-v_i(f(g')))},\label{eq:base-orderable-sum-i}
\end{align}
where the second transition follows because 
\begin{align*}
v_i(A_j)-v_i(A_i)  = (v_i(A^*_j) - v_i(A^*_i)) + (v_i(A_j \setminus A^*_j) - v_i(A_i \setminus A^*_i)),
\end{align*}
and 
\[
v_i(A_j \setminus A^*_j) - v_i(A_i \setminus A^*_i) = \sum\nolimits_{g \in A_j \setminus A^*_j} v_i(g)-v_i(f(g)) \le 0;
\]
and the third transition follows due to the assumed violation of \half-EF1. 

Now, we can observe that
\begin{align}
    \sum_{g \in A^*_j}{(v_j(g) - v_j(f(g)))} = v_j(A^*_j) - v_j(A^*_i) \le v_j(A^*_j) \le v_j(A_j).\label{eq:base-orderable-sum-j}
\end{align}

Note that for every $g \in A^*_j$, not only is $v_i(g) > v_i(f(g))$ (due to the definition of $A^*_j$), but also $v_j(g) > v_j(f(g))$. If this were not true, then swapping $g$ and $f(g)$ between $A_i$ and $A_j$, which yields a complete and feasible allocation due to $A_i \cup \set{g^*} \setminus \set{f(g^*)}$ and $A_j \setminus \set{g^*} \cup \set{f(g^*)}$ also being bases of $(M,\calI)$, would be a Pareto improvement over $A$. Choosing $g^* \in \argmin_{g \in A^*_j} \frac{v_j(g) - v_j(f(g))}{v_i(g) - v_i(f(g))}$, we have
\begin{align*}
    \frac{v_j(A_j)}{v_i(A_i) + v_i(g^*) - v_i(f(g^*))}
    & > \frac{\sum_{g \in A^*_j}{(v_j(g) - v_j(f(g)))}}{\sum_{g \in A^*_j}{(v_i(g) - v_i(f(g)))}} \\
    & \geq \frac{v_j(g^*) - v_j(f(g^*))}{v_i(g^*) - v_i(f(g^*))} \\
    & = \frac{v_j(A_j) - v_j(A_j \setminus \set{g^*} \cup \set{f(g^*)})}{v_i(A_i \cup \set{g^*} \setminus \set{f(g^*)}) - v_i(A_i)},
\end{align*}
which can be rearranged to get:
\begin{align*}
    v_i\left(A_i \cup \set{g^*} \setminus \set{f(g^*)}\right) \cdot v_j\left(A_j \setminus \set{g^*} \cup \set{f(g^*)}\right) > v_i(A_i) \cdot v_j(A_j).
\end{align*}
Consider allocation $B$ where $B_i = A_i \cup \set{g^*} \setminus \set{f(g^*)}$, $B_j = A_j \setminus \set{g^*} \cup \set{f(g^*)}$, and $B_k = A_k$ for $k \in N\setminus\set{i,j}$. Note that $B$ is also a complete allocation. If either $i$ or $j$ had zero utility under $A$, then $B$ gives a positive utility to strictly more agents ($|P(B)| > |P(A)|$), and if both had a positive utility under $A$, then $B$ has a strictly higher product of positive agent utilities ($\NW(B) > \NW(A)$). In either case, it contradicts $A$ being a complete MNW allocation. This proves that $A$ must be \half-EF1 too. 
\end{proof}

\subsection{Beyond Matroidal Constraints}\label{sec:lower-bounds}

One may be tempted to apply the trick from \Cref{thm:matroid} of reducing to unconstrained instances for even non-matroidal constraints. Even if this induces valuations from the broader class of subadditive valuations,\footnote{Valuation $v$ is subadditive if $v(S\cup T) \le v(S)+v(T)$ for all $S,T \subseteq M$.} one can use a recent result by \cite{BS24} to establish the existence of EF1 and \half-PO allocations. However, even for very simple non-matroidal constraints, it is easy to see that the induced valuations would not even be subadditive.

For example, consider partition matroids with lower bounds, where the number of goods allocated to every agent from each category $C$ must be in between a lower bound $\ell_C$ and an upper bound $h_C$ for category $C$. Now, an agent's value for a bundle of goods can stay zero while the bundle contains fewer than $\ell_C$ goods from any category $C$, and suddenly jump to a positive value once it contains at least $\ell_C$ goods from every category $C$, violating subadditivity. 

Nonetheless, for the case of partition matroids with lower bounds, we develop a novel approach of \emph{alternate worlds construction} to establish that complete MNW allocations are \half-EF1 and PO. The following lemma lays out the approach. 

\begin{restatable}{lemma}{nashaltworlds}\label{lem:nashaltworlds}
    Let $\set{\allocs_1, \allocs_2, \dots, \allocs_k}$ be a collections of sets of feasible allocations (``alternate worlds''), defined over the same set of goods $M$. Fix any $\alpha \in [0,1]$. If every MNW allocation among $\allocs_t$ is $\alpha$-EF1 and PO for each $t \in [k]$, then every MNW allocation among $\allocs = \cup_{t \in [k]}{\allocs_t}$ will be $\alpha$-EF1 and PO.
\end{restatable}
\begin{proof}
    Consider any allocation $A$ that is MNW among $\allocs = \cup_{t \in [k]}{\allocs_t}$. The fact that it is PO among $\allocs$ follows through the same argument as for the unconstrained case (see the proof of \Cref{thm:base-orderable}). The fact that $A$ is $\alpha$-EF1 follows from the fact that $A$ must be an MNW allocation among $\allocs_t$ for some $t \in [k]$ (i.e., in one of the alternate worlds). For contradiction, assume that this is false. Since $\allocs = \cup_{t \in [k]}{\allocs_t}$, we must have $A \in \allocs_t$ for some $t \in [k]$. Then, the fact that $A$ is not an MNW allocation among $\allocs_t$ implies that there exists another allocation $B \in \allocs_t$ which either gives positive utility to strictly more agents or gives positive utility to the same set of agents as $A$ but has a higher product of their utilities. Since $B \in \allocs$, this would also contradict $A$ being an MNW allocation among $\allocs$. 
\end{proof}

We can apply \Cref{lem:nashaltworlds} to \Cref{thm:matroid,thm:base-orderable} to obtain \half-EF1 and PO for a broad class of non-matroidal constraints obtained as ``unions of matroid constraints''. Note that, usually, the union of matroids $(M,\calI_1)$ and $(M,\calI_2)$ is defined as $(M,\calI)$, where $\calI = \set{I_1 \cup I_2 \mid I_1 \in \calI_1, I_2 \in \calI_2}$, which is always a matroid~\citep{oxley2011matroid}. The union we are referencing is different: if $\allocs_1$ and $\allocs_2$ are the sets of feasible allocations induced by matroid constraints $(M,\calI_1)$ and $(M,\calI_2)$ respectively, then the ``union'' of these constraints would have the set of feasible allocations $\allocs = \allocs_1 \cup \allocs_2$, which may not be induced by any matroid constraint. 
\begin{corollary}\label{cor:union-matroids}
    Given a set of feasible allocations $\allocs = \cup_{t=1}^k \allocs_k$, where $\allocs_t$ is the set of feasible allocations under a matroid constraint for each $t \in [k]$, every MNW allocation from $\allocs$ is \half-EF1 and PO. If $\allocs_t$ is the set of complete and feasible allocations under a base-orderable matroid constraint for each $t \in [k]$, then every complete MNW allocation from $\allocs$ is \half-EF1 and PO. 
\end{corollary}

Next, we show that for \emph{complete} allocations, the non-matroidal constraint imposed by a partition matroid with lower bounds can be represented as a union of partition matroid constraints (which are base-orderable). 

\begin{restatable}{lemma}{lowerboundaltworlds}\label{lem:lowerboundaltworlds}
    Let $\allocs$ be the set of complete and feasible allocations under a partition matroid constraint with lower bounds. Then, $\allocs = \cup_{t=1}^k \allocs_t$, where $\allocs_t$ is the set of complete and feasible allocations under a partition matroid constraint for each $t \in [k]$. 
\end{restatable}
\begin{proof}
    Let $\calC$ be the set of categories with associated lower bounds $(\ell_C)_{C \in \calC}$ and upper bounds $(h_C)_{C \in \calC}$ in a partition matroid constraint with lower bounds. Let $\allocs$ be the set of complete and feasible allocations under this constraint. Note that in any allocation in $\allocs$, from each category $C$, \emph{some} subset of $n \cdot \ell_C$ goods are divided evenly among the agents (so each agent receives $\ell_C$ goods), and the rest are divided so that each agent receives at most $h_C-\ell_C$ of them. This inspires the following construction.  
    
    We create $\prod_{C \in \calC} \binom{|C|}{n \cdot \ell_C}$ many partition matroids, one for each possible way of selecting $Q_C \subseteq C$ of size $n \cdot \ell_C$ from each $C \in \calC$. In the partition matroid corresponding to a given choice of $Q = (Q_C)_{C \in \calC}$, there are two categories for every category $C \in \calC$: $C^1_Q = Q_C$ with an upper bound of $h_{C^1_Q} = \ell_C$ and $C^2_Q = C \setminus C^1_Q$ with an upper bound of $h_C-\ell_C$. The key observation is that under any \emph{complete} and feasible allocation subject to this partition matroid, goods in $C^1_Q$ will be divided exactly evenly (so that each agent receives $\ell_C$ of them) while goods in $C^2_Q$ will be divided so that each agent receives at most $h_C-\ell_C$ of them. This immediately yields that any complete and feasible allocation under any such partition matroid is a complete and feasible allocation in $\allocs$. 

    To see the converse, take any $A \in \allocs$. For each $C \in \calC$, define $Q^1_C = \cup_{i \in N} Q_{C,i}$, where $Q_{C,i}$ is an arbitrary subset of $A_i \cap C$ of size $\ell_C$. Note that the $n \cdot \ell_C$ goods in $Q^1_C$ are divided evenly among the agents, with each receiving exactly $\ell_C$ of them. Also, since $|A_i \cap C| \le h_C$, each agent must receive at most $h_C - \ell_C$ goods from $Q^2_C = C \setminus Q^1_C$. Thus, $A$ is a complete and feasible allocation under the partition matroid corresponding to $Q = (Q_C)_{C \in \calC}$.
\end{proof}

Plugging in \Cref{lem:lowerboundaltworlds} into \Cref{cor:union-matroids} yields \half-EF1 and PO of complete MNW allocations subject to partition matroids with lower bounds. It is not clear if we can decompose the set of all (possibly incomplete) feasible allocations subject to partition matroids with lower bounds in the same manner, which would allow us to directly apply \Cref{cor:union-matroids} for MNW allocations, as our proof of \Cref{lem:lowerboundaltworlds} crucially uses completeness. Nonetheless, we show that \half-EF1 and PO guarantees can be extended to MNW allocations too.

\begin{restatable}{theorem}{partitionLB}\label{thm:partitionLB}
    Under any partition matroid constraint with lower bounds, every MNW and complete MNW allocation is \half-EF1 and PO. 
\end{restatable}
\begin{proof}
    The result for complete MNW allocations (with Pareto optimality only within the set of all complete and feasible allocations) follows from \Cref{lem:lowerboundaltworlds} and \Cref{cor:union-matroids}. 
    
    Let $\calC$ be the set of categories with associated lower bounds $(\ell_C)_{C \in \calC}$ and upper bounds $(h_C)_{C \in \calC}$ in any partition matroid constraint with lower bounds, and let $A$ be any (possibly incomplete) MNW allocation in that instance. The fact that $A$ is PO (within the set of all feasible allocations) follows exactly as in the proof of \Cref{thm:matroid} (as it does not use the matroid property of the constraint). 
    
    Let us now show that $A$ is \half-EF1. First, we extend $A$ to a complete allocation $\tilde{A}$ by arbitrarily allocating the unallocated goods while keeping the allocation feasible. This is possible for partition matroids with lower bounds because our assumption of $\ell_C \le \floor{\nicefrac{|C|}{n}} \le \ceil{\nicefrac{|C|}{n}} \le h_C$ for each $C \in \calC$ means that whenever there is an unallocated good from any category, there must be an agent who can feasibly receive it.

    Note that $\tilde{A}$ must be a complete MNW allocation. This uses the same argument as in the proof of \Cref{thm:matroid} too: if $\tilde{A}$ was not a complete MNW allocation, we could find a complete MNW allocation $B$ that would either give positive utility to more agents or to the same set of agents but with a higher Nash welfare than that of $\tilde{A}$. Since $\tilde{A}$ and $A$ induce the same utilities, this would also contradict $A$ being an MNW allocation. Hence, \half-EF1 of $\tilde{A}$ follows from above. 

    To see that $A$ is \half-EF1, fix any pair of agents $i,j \in N$. We want to show that either $A_j = \emptyset$ or $2v_i(A_i) \ge v_i(A_j)-v_i(g)$ for some $g \in A_j$. From \half-EF1 of $\tilde{A}$, we have that either $\tilde{A}_j = \emptyset$ or $2v_i(\tilde{A}_i) \ge v_i(\tilde{A}_j)-v_i(g^*)$ for some $g^* \in \tilde{A}_j$. In the former case, $A_j \subseteq \tilde{A}_j$ would imply $A_j = \emptyset$. And in the latter case, we have 
    \[
    2v_i(A_i) = 2v_i(\tilde{A}_i) \ge v_i(\tilde{A}_j \setminus \set{g^*}) \ge v_i(A_j \setminus \set{g^*}),
    \]
    where the first transition holds because $\tilde{A}$ cannot be a Pareto improvement over $A$ ($A$ is PO) and the last transition holds again because $A_j \subseteq A^*_j$.
\end{proof}

\subsection{Implications}\label{sec:implications}

\paragraph{Implication for partition matroids.} For partition matroids (cardinality constraints), \cite{biswas2018fair} show that an EF1 allocation always exists, whereas \Cref{thm:matroid} shows that MNW allocations achieve \half-EF1 and PO, sacrificing some fairness for added efficiency.  

\paragraph{Implications for balancedness.} Balancedness is the special case of partition matroids with lower bounds, where all bundle sizes are almost equal: $|A_i| \in \set{\floor{\nicefrac{m}{n}},\ceil{\nicefrac{m}{n}}}$ for all $i \in N$. It is known that round robin achieves exact EF1 subject to balancedness~\citep{CKMP+19}, but we provide an example in \Cref{app:round-robin-inefficient} showing that it is inefficient.
In contrast, \Cref{thm:partitionLB} shows that MNW allocations subject to balancedness are \half-EF1 and PO, sacrificing some fairness in favor of increased efficiency, as in the case of partition matroids. To the best of our knowledge, this is the first result about balanced MNW allocations.

\paragraph{Implication for goods with copies.} As this is a special case of partition matroids, the same implication holds for MNW allocations. However, for this case, \cite{GHLT24} argue that a more suitable notion of fairness is \efwc, which is a strengthening of EF1. This is because a chore division instance with additive costs can be modeled as a goods with copies instance with additive values and $n-1$ copies of a good corresponding to every chore: \efwc in the goods with copies instance would imply EF1 for the underlying chore division instance, while EF1 for the goods with copies instance would not. 

\begin{definition}[(Approximate) Envy-Freeness Up to One Good Without Commons (\efwc)]
For $\alpha \in [0,1]$, an allocation $A \in \allocs$ is $\alpha$-approximate \emph{envy-freeness up to one good without commons} ($\alpha$-\efwc) if, for every pair of agents $i,j \in N$, either $v_i(A_i) \geq \alpha \cdot v_i(A_j)$ or there exists a category $C \in \calC$ and a good $g \in C$ such that $A_i \cap C = \emptyset$ and $v_i(A_i) \ge \alpha \cdot v_i(A_j \setminus \set{g})$. In words, agent $i$ should not envy agent $j$, up to a factor of $\alpha$, after excluding some good from agent $j$'s bundle \emph{that agent $i$ does not also have a copy of}. When $\alpha=1$, we simply call it an \efwc allocation. 
\end{definition}

The fact that we can only exclude a good from $A_j$ of which agent $i$ does not already have a copy makes $\alpha$-\efwc stronger than $\alpha$-EF1. We already know from \Cref{thm:matroid,thm:base-orderable} that every MNW and complete MNW allocation for a goods with copies instance achieves \half-EF1. We can strengthen this to show that they in fact achieve \half-\efwc.

First, we need the following lemma, showing that for goods with copies (in fact, even for partition matroids), completeness comes without a price for MNW allocations.

\begin{restatable}{lemma}{partitionMNWComplete}\label{lem:partitionMNWComplete}
    Under any partition matroid constraint, every complete MNW allocation is an MNW allocation (among the set of all feasible allocations). 
\end{restatable}
\begin{proof}
    As part of the proof of \Cref{thm:partitionLB}, we proved that every MNW allocation under a partition matroid constraint with lower bounds can be extended to a complete MNW allocation. This proof is similar, but in the opposite direction. 
    
    For contradiction, assume this is false, and there exists some complete MNW allocation $A$ for a instance constrained by the partition matroid constraints $\calC$ that is not MNW over all allocations. In this case, we know that there can be no complete allocations that maximize Nash Welfare for this instance, or that would contradict the fact that $A$ is complete MNW. Therefore, all MNW allocations must be incomplete, let $A'$ be the MNW allocation that minimizes the number of unallocated goods.

    Since we know that $\calC$ admits at least $1$ complete allocation, it must be the case that for each constraint $C_t \in \calC$, we have that $h_t \geq \ceil{|C_t|/n}$. Since $A'$ is not complete, there must be some constraint $C' \in \calC$ and some good $g \in C'$ such that $g \not \in A'_i$ for all $i \in N$. Further, it must be the case that for some agent $j \in N$, $|A'_j \cap C'| < \ceil{|C_t|/n}$. If not then there would be $\ceil{|C_t|/n}n \geq |C_t|$ allocated goods from $C_t$, contradicting the existence of $g$. Finally, note that because of this, $A'_j \cup \set{g}$ must be a feasible bundle, and we must have that $v_j(A'_j) \leq v_j(A'_j \cup \set{g})$. Consider what happens if we replace $A'_j$ with $A'_j \cup \set{g}$, this allocation would not decrease the number of agents with positive utility, or the product of positive utilities, meaning it must also be an MNW allocation. However, it will strictly decrease the number of unallocated goods, causing a contradiction, and proving that $A$ must be MNW over all allocations.
\end{proof}

Now, we are ready to prove that MNW allocations in fact get \half-\efwc for goods with copies.

\begin{restatable}{theorem}{copies}\label{thm:copies}
    In the case of goods with copies, every MNW and complete MNW allocation is \half-\efwc and PO.
\end{restatable}
\begin{proof}
    For some instance constrained by a partition matroid $\mathcal{M} = (M,\mathcal{I})$, and any agent $i \in N$, let $B_i: 2^M \rightarrow \mathcal{I}$ be the function that maps any set $S \subseteq M$ to $i$'s most preferred feasible allocation $I$ such that $I \subseteq S$.

    First notice that in any MNW allocation $A$ in a matroid constrained instance, and for any pair of agents $i,j \in N$, either $A_j = \emptyset$, or there exists some good $g \in A_j$ such that  $2v_i(A_i) \geq v_i(B_i(A_i \cup A_j \setminus \set{g}))$. One can see that this condition is derived from the definition of ``Marginal EF1'' that \cite{CKMP+16} prove holds for unconstrained instances with submodular valuation functions, for the purposes of this proof we will refer to it as ``Constrained Marginal EF1''. In the proof of \Cref{thm:matroid}, we show that every MNW allocation over a matroid constraint corresponds to an allocation over an unconstrained setting where each agent $i \in N$ has the valuation function $\tv_i$, equaling the the ``weighted rank function'' of the matroid with their original valuation function $v_i$ being used as the weights. It is known that this unconstrained allocation will be Marginal EF1, and for every $i,j \in N$, we will have $\tv_i(A_j) = v_i(A_j)$. This directly implies that the MNW allocation over the original constrained instance will be Constrained Marginal EF1.

    It is clear that if an allocation is constrained marginal EF1, then it will also be \half-EF1. What we will show next is that if the property holds for an allocation $A$ over a goods with copies instance, then $A$ will be \half-\efwc.

    For contradiction, assume that this is false, and for some goods with copies instance constrained by $\calC$, there exists an allocation $A$ that is Constrained Marginal EF1 but not \half-\efwc.

    Since we know that $A$ is Constrained Marginal EF1, we know there exists some constraint $C \in \calC$, and some $g \in A_j \cap C$, such that $2v_i(A_i) \geq v_i(B_i(A_i \cup A_j \setminus g))$. We will analyze two cases:

    Case 1: $|A_i \cap C| = 0$. In this case, \half-\efwc directly follows from the fact that $2v_i(A_i) \geq v_i(B_i(A_i \cup A_j \setminus \set{g}))$ implies $2v_i(A_i) \geq v_i(A_j \setminus \set{g})$.

    Case 2: $|A_i \cap C| = 1$. In this case, by the way that goods with copies constraints are structured, it must be true that $B_i(A_i \cup A_j \setminus \set{g}) = B_i(A_i \cup A_j)$. We know that agent $i$ has the same valuation for all goods in $C$. We also know that there can only be $1$ good from $C$ in any feasible bundle, and that $A_i \cup A_j$ contains at least two goods from $C$ ($g$ as well as some good $g' \in A_i$). Since $A_i \cup A_j$ has multiple goods from $C$ that both have equivalent value, it is easy to see that removing one of them will have no affect on the best possible bundle $i$ can construct using the goods. Therefore, in this case, Constrained Marginal EF1 implies that $2v_i(A_i) \geq v_i(B_i(A_i \cup A_j)) \geq v_i(A_j)$, which gives us \half-\efwc.

    These two cases exhaust all possible scenarios (either $A_i$ has a copy from $C$ or it does not), giving us a contradiction.
    
    This proves that for goods with copies instances, all MNW allocations will be \half-\efwc. From \Cref{lem:partitionMNWComplete}, we can also conclude that the set of all complete MNW allocations for a given instance will be a subset of the MNW allocations, and they will therefore be \half-\efwc as well.
\end{proof}

Because the primary motivation of goods with copies is that it can be used to model chore division as a special case, this is one example where free disposal of goods is not possible (otherwise, some chores may be assigned to multiple agents). Hence, it is useful that \Cref{thm:copies} handles not only MNW allocations, but also complete MNW allocations. 

To the best of our knowledge, this is the first known approximation for \efwc under general additive preferences. \cite{GHLT24} leave the existence of \efwc open, only proving it for special cases such as that of leveled preferences. The fact that MNW allocations, in addition to being \half-\efwc, also satisfy PO is an added bonus.

However, there is one significant downside to \Cref{thm:copies}. When a chore division instance is modeled as a goods with copies instance, while \efwc for the latter implies EF1 for the former, even \half-\efwc for the latter does not imply any reasonable approximation of EF1 for the former, as the following example shows.

\begin{example}
    Consider the following instance with $3$ agents and $3$ goods labeled $g_1$ through $g_3$, each with $2$ copies.
    \begin{table}[ht]
        \centering
        \begin{tabular}{c | c c c}
            \toprule
              & $g_1$ & $g_2$ & $g_3$ \\
            \midrule
            Agent $1$ & 1     & 1/100 & 1/100 \\
            Agent $2$ & 1     & 1     & 1     \\
            Agent $3$ & 1     & 1     & 1     \\
            \bottomrule
        \end{tabular}
    \end{table}

    The MNW allocation $A$ (unique, up to symmetry between agents $2$ and $3$) is given by $A_1 = \set{g_1}$, $A_2 = \set{g_1,g_2,g_3}$, and $A_3 = \set{g_2,g_3}$. One can verify that this allocation is \half-\efwc. However, for the underlying chore division instance, where the values of each agent for each chore is the negation of the value for the corresponding good given above, the corresponding allocation $A'$ of chores is given by $A'_1 = \set{g_2,g_3}$, $A'_2 = \emptyset$, and $A'_3 = \set{g_1}$. This allocation provides zero approximation of EF1 because we have $v_1(A'_1 \setminus \set{g_2}) = v_1(A'_1 \setminus \set{g_3}) = -1/100$ but $v_1(A'_2) = 0$.
\end{example}

Recently, \cite{garg2024halfef1chores} show that \half-EF2 and PO can be achieved for chore division through a Fisher market based algorithm. It would be interesting to extend this to goods with copies.

\paragraph{Search for EF1+PO allocations.} The primary goal of our work is to initiate a systematic study of the existence of \emph{fair and efficient} allocations under feasibility constraints. The ultimate goal of seeking (exactly) EF1 and PO allocations still remains a major open question in all cases.

\begin{open}
    Does an allocation that is EF1 and PO always exist subject to an arbitrary matroid constraint and heterogeneous additive valuations? What about just the special case of goods with copies?
\end{open}

It is worth remarking that subject to matroid constraints (resp., budget constraints), MNW allocations achieve \half-EF1 (resp., \sfrac{1}{4}-EF1) with PO (\Cref{thm:matroid} and \cite{wu2021budget}), while the open question is whether exact EF1 (resp., \half-EF1) with PO is achievable. Curiously, in both cases, MNW allocations achieve half the fairness level of the known upper bound subject to PO. 

In fact, even the existence of an EF1 allocation remains open beyond partition (resp., base-orderable) matroids for heterogeneous (resp., identical) additive valuations~\citep{biswas2018fair,biswas2019matroid}, and the existence of an \efwc allocation remains open for goods with copies.

\subsection{Impossibility of Stronger Guarantees}\label{sec:stronger}

In this section, we explore whether the fairness and efficiency guarantees established in \Cref{thm:matroid} (and in the subsequent results above) can be strengthened. We explore several directions, establishing a negative result in almost all cases but also identifying interesting open questions. 

\subsubsection{Tightness of the Fairness Guarantee of MNW}

First, we show that MNW does not achieve a better guarantee than $\half$-EF1 under matroid constraints.

\begin{restatable}{theorem}{mnwtight}\label{thm:mnw-tight}
For any $\epsilon > 0$, there exist instances with partition matroid constraints and the balancedness constraint in which no MNW (or complete MNW) allocation is $(\half + \epsilon)$-EF1.
\end{restatable}
\begin{proof}
Let $k \in \bbN$. Consider the partition-matroid-constrained instance with $n = 2$ agents and $m = 2k$ goods, with the goods being partitioned into two groups, $S = \set{g_1,\dots,g_k}$ and $T = \set{g_{k+1},\ldots,g_{2k}}$. For each good $g_s \in S$, we have $v_1(g_s) = v_2(g_s) = 1$. For each good $g_t \in T$, we have $v_1(g_t) = 0$ and $v_2(g_t) = 1/2$. The partition matroid constraints dictate that no agent can receive more than $k$ goods from the entire set of goods $M = S \cup T$. 

We argue that the only MNW allocation is the complete allocation $A$ with $A_1 = S$ and $A_2 = T$. To see this, consider any allocation $A$. Let $x = v_1(A_1) = |A_1 \cap S|$. Then, $v_2(A_2) \le (k-x)+\frac{1}{2} x$, as agent $2$ can receive at most $k$ goods in total, and her utility is maximized by receiving the remaining $k-x$ goods of $S$ along with some $x$ goods from $T$. Hence, $\NW(A) \le x \cdot \left(k-\frac{1}{2}x\right)$. The right hand side is uniquely maximized at $x=k$, resulting in the allocation $A$ with $A_1 = S$ and $A_2 = T$ uniquely attaining the highest Nash welfare of $\frac{k^2}{2}$. 

In this allocation, agent $1$ is not envious, but agent $2$ has $v_2(A_2) = \frac{1}{2}k$ whereas $v_2(A_1 \setminus \set{g}) = k-1$ for any $g \in A_1$. Hence, this allocation is $\frac{k}{2(k-1)}$-EF1, i.e., $(\frac{1}{2} + \frac{1}{2k-2})$-EF1. For any $\epsilon > 0$, one can choose $k > 1+\frac{1}{2\epsilon}$ to ensure that the unique MNW allocation is not $(\half + \epsilon)$-EF1.

Since the unique MNW allocation in this instance is in fact balanced, changing the partition matroid constraint (each agent receives at most $k$ goods) to the balancedness constraint (each agent receives exactly $k$ goods) results in the same allocation being the unique MNW allocation, thus establishing tightness of the $\half$-EF1 guarantee with respect to the balancedness constraint too. 
\end{proof}

The tightness result above clearly extends to all families of constraints more general than partition matroid constraints (in particular, to laminar and base-orderable matroid constraints) or balancedness (in particular, to partition matroids with lower bounds). Thus, the only family of constraints we study that is not covered is goods with copies, for which we provide a similar example in \Cref{app:lowerboundcopies}.

\subsubsection{Strengthening Fairness Criteria}

One common strengthening of EF1 is ``Stochastically Dominant EF1'' (SD-EF1). An allocation satisfies SD-EF1 if it can be inferred that the allocation is EF1 by only looking at the agents' ordinal preferences over the goods. We provide the formal definition of SD-EF1 below:

Define $\succeq_i$ (resp., $\succ_i$) as the weak (resp., strict) ordering over the goods in $M$ induced by $v_i$, where, for all $g,g' \in M$, $g \succeq_i g'$ if and only if $v_i(g) \geq v_i(g')$ and $g \succ_i g'$ if and only if $v_i(g) > v_i(g')$.

\begin{definition}[SD-EF1]
    An allocation $A$ of a set of goods $S$ is stochastic-dominance envy-free up to one good (SD-EF1) if for all $i,j \in N$ with $A_j \neq \emptyset$, there exists a $g^* \in A_j$ such that for all $g \in S$, $\card{\set{g' \in A_i : g' \succeq g}} \ge \card{\set{g' \in A_j \setminus \set{g^*} : g' \succeq g}}$.
\end{definition}

In unconstrained instances, SD-EF1 is know to be achievable through the ``round-robin'' mechanism. Since round-robin always produces a balanced allocation, this means that SD-EF1 is achievable under balancedness constraints as well. However, under more general partition matroid constraint, we can show that SD-EF1 is no longer able to be guaranteed.

\begin{restatable}{theorem}{sdefone}\label{thm:sdef1}
    Under partition matroid constraints, SD-EF1 cannot be guaranteed.
\end{restatable}
\begin{proof}
    Consider the following instance, in which there are two agents, and eight goods partitioned into $4$ subsets of $2$ goods, represented by the red groupings. We place an upper-bound of $1$ on each grouping, meaning that each agent must receive exactly $1$ good from each group.

    \begin{center}
    \renewcommand{\arraystretch}{1.2}
    \begin{tabular}{c | c c c c c c c c}
        \toprule
               & 
               \tikz[baseline, remember picture]\node (g1) {$g_1$}; & 
               \tikz[baseline, remember picture]\node (g2) {$g_2$}; & 
               \tikz[baseline, remember picture]\node (g3) {$g_3$}; & 
               \tikz[baseline, remember picture]\node (g4) {$g_4$}; & 
               \tikz[baseline, remember picture]\node (g5) {$g_5$}; & 
               \tikz[baseline, remember picture]\node (g6) {$g_6$}; & 
               \tikz[baseline, remember picture]\node (g7) {$g_7$}; & 
               \tikz[baseline, remember picture]\node (g8) {$g_8$}; \\
        \midrule
         Agent $1$ & 8 & 4 & 7 & 5 & 6 & 1 & 3 & 2 \\
         Agent $2$ & 6 & 2 & 8 & 5 & 7 & 3 & 4 & 1 \\
         \bottomrule
    \end{tabular}
    \end{center}
    
    \begin{tikzpicture}[overlay, remember picture]
        \draw[red, thick, rounded corners] ($(g1.north west)+(0,0)$) rectangle ($(g2.south east)+(0,-0)$);
        \draw[red, thick, rounded corners] ($(g3.north west)+(0,0)$) rectangle ($(g4.south east)+(0,-0)$);
        \draw[red, thick, rounded corners] ($(g5.north west)+(0,0)$) rectangle ($(g6.south east)+(0,-0)$);
        \draw[red, thick, rounded corners] ($(g7.north west)+(0,0)$) rectangle ($(g8.south east)+(0,-0)$);
    \end{tikzpicture}

    In any SD-EF1 allocation $A$ over this instance, it follows from the definition that each agent must get at least $1$ of their top $2$ goods. This means that $A_1$ must contain one of $\set{g_1,g_3}$ and $A_2$ must contain at least one of $\set{g_3,g_5}$. Similarly, it must be the case that each agent receives at least $2$ of their top $4$ goods. This means that both $A_1$ and $A_2$ must contain exactly $2$ of $\set{g_1,g_3,g_4,g_5}$.

    Notice that under these restrictions, in any SD-EF1 allocation where agent $1$ receives $g_3$, they cannot receive $g_5$ (or they would have both of agent $2$'s top $2$ goods) or $g_4$ (due to the constraints). Since the allocation over the top $4$ goods must be balanced (each agent receiving exactly two of them), the only possible allocation over $\set{g_1,g_3,g_4,g_5}$ in this scenario would be $(\{g_1,g_3\},\{g_4,g_5\})$. Using the same logic, when agent $2$ receives $g_3$, the only possible allocation is $(\{g_1,g_4\},\{g_3,g_5\})$. Since one of the agents must be given $g_3$, these are the only two ways that the top $4$ goods can be allocated. Note that in both of these allocations, agent 1 gets $g_1$ and agent 2 gets $g_5$.

    Finally, consider how the remaining goods $\{g_2,g_6,g_7,g_8\}$ must be allocated to guarantee SD-EF1. Each agent must receive at least $3$ of their top $6$ goods, and since we know that each agent has exactly $2$ of their top $4$ goods, that means that agent $1$ must receive at least one of $\{g_2,g_7\}$, and agent $2$ must receive at least one of $\{g_6,g_7\}$. Since we know that agent 1 must be allocated $g_1$, the matroid constraints say they cannot receive $g_2$, so they must receive $g_7$. Similarly, agent $2$ is known to have $g_5$, so they cannot receive $g_6$, which means they must also receive $g_7$. Both agents cannot be simultaneously allocated $g_7$, hence, in this instance, there is no complete and feasible allocation that is SD-EF1.
\end{proof}

Notably, this settles the question of SD-EF1 existence for all of the main classes of constraints studied in this paper, except for goods with copies. We leave whether SD-EF1 can be achieved for goods with copies as an open question.

\begin{open}
    Does there always exist an SD-EF1 allocation of goods with copies?
\end{open}

\subsubsection{Strengthening Efficiency Criteria}

In this section, we first consider a strengthening of PO, which we term \poplus. Intuitively, an allocation $A \in \allocs$ is \poplus if it is not Pareto dominated by any unconstrained allocation, i.e. any possible allocation that could be constructed from a base set of goods, regardless if it is in $\allocs$. We define the set of complete unconstrained allocations analogously.

\begin{definition}[Unconstrained Pareto Optimality (\poplus)]
    Let $\mathcal{U}$ be the set of unconstrained allocations. An allocation $A \in \allocs$ is \poplus if there is no other allocation $B \in \mathcal{U}$ such that $v_i(B_i) \ge v_i(A_i)$ for each $i \in N$ and at least one inequality is strict. In words, no other allocation, feasible or infeasible, should be able to make an agent happier without making any agent less happy.
\end{definition}

For goods with copies, and hence, for partition and laminar matroid constraints, \poplus is easily seen to be too strong, even in the absence of any other desiderata. 

\begin{example}
    Consider an goods with copies instance with $3$ agents, and $3$ goods with $2$ copies each. Agents have the following valuations for the goods, where $\epsilon < \frac{1}{2}$ is some very small value:
    \begin{table}[ht]
        \centering
        \begin{tabular}{c | c c c}
            \toprule
                   & $g_1$ & $g_2$ & $g_3$ \\
            \midrule
             Agent 1 & 1     & $\epsilon$     & $\epsilon$     \\
             Agent 2 & $\epsilon$     & 1     & $\epsilon$     \\
             Agent 3 & $\epsilon$     & $\epsilon$     & 1     \\
             \bottomrule
        \end{tabular}
    \end{table}

    If we ignored the restrictions on feasibility enforced by the goods with copies constraints, then clearly the most efficient way to allocate the goods would be to assign both copies of $g_i$ to $a_i$ for each $i \in \set{1,2,3}$. This would result in each agent having a valuation of $2$. However, under the goods with copies constraints, each agent can only get a single copy of the one good they find valuable, so in any feasible allocation, there would be no agent with a utility higher than $1 + 2\epsilon < 2$.
\end{example}

However, the balanced setting is a curious special case. A result by \cite{CFS17} implies the following.

\begin{theorem}[\citeauthor{CFS17}~\citeyear{CFS17}]
    In the balanced case, when each agent has a strictly positive valuation for each good, there always exists a feasible allocation that is \poplus.
\end{theorem}

This raises the interesting question, in the balanced case where all agents have strictly positive utilities, does there always exist an allocation that is EF1+\poplus? It turns out this is not true. As we show in the below theorem, it is not possible to guarantee any approximation of EF1 along with \poplus in the balanced case.

\begin{restatable}{theorem}{thmpoplus}\label{thm:poplus}
Under the balancedness constraint, it is impossible to guarantee $\alpha$-EF1 and \poplus for any $\alpha \in (0,1]$.
\end{restatable}
\begin{proof}
    For contradiction, assume this is false, and for some $\alpha \in (0,1]$, every balanced instance admits an allocation that is \poplus and $\alpha$-EF1. 

    Let $k = \lceil 2 / \alpha \rceil$. Consider the following balanced-constrained instance, with $2$ agents and $m = 2(k^2 + k)$ goods. 
    The $m$ goods are broken up into two types, a set $L$ of $m/2$ low-valued goods, for which $v_1(g) = v_2(g) = 1$ for all $g \in L$, and a set $H$ of high-valued goods, for which $v_1(g) = k+1$ and $v_2(g) = k$ for each $g \in H$. Because this the allocation must be balanced, in every feasible instance both agents will receive exactly $k^2 + k$ goods.

    By our assumption, we know that there exists an allocation for this instance that is $\alpha$-EF1 and \poplus, call this allocation $A$. 

    We first note that any instance in which agent $1$ receives at least $k$ goods from $L$ will not be \poplus. To see this, consider what happens if agent $1$ has $k$ goods from $L$ and agent $2$ has $1$ good from $H$. It would be a Pareto improvement for the agents to swap these goods, as agent $1$'s utility would increase by $1$, and agent $2$'s utility would remain the same. 
    
    Also note that if agent $1$ has at most $k - 1$ goods from $L$, then due to the balancedness constraint that they must have $m/2 = k^2 + k$ total goods, agent $1$ must have at least $k^2 + 1$ goods from $H$. Since $A_2 = M \setminus A_1$, the inverse of this argument would be that agent $2$ has at most $k - 1$ goods from $H$ and at least $k^2 + 1$ goods from $L$.
    Therefore,
    for agent $2$, we have
    \begin{align*}
    v_2(A_1) &\ge (k^2 + 1) \cdot k + (k - 1) \cdot 1 > k^3 + k, 
    \\
    v_2(A_2) &\le (k - 1) \cdot k + (k^2 + 1) \cdot 1 \le 2k^2,
    \end{align*}
    Since $A$ is $\alpha$-EF1, we know the following statement must be true
    \[
    v_2(A_2) \geq \alpha \cdot ( v_2(A_1) - k).
    \]
    Substituting the bounds from above, we have
    \[
    2k^2 > \alpha k^3 \iff 2/\alpha > k,
    \]
    which contradicts our assumption that $k = \lceil {2}/{\alpha}\rceil$. Therefore, $A$ cannot be $\alpha$-EF1 and \poplus.
\end{proof}

\section{Best-of-Both-Worlds Guarantees}\label{sec:bobw}

Next, we present randomized allocations that achieve desirable fairness and efficiency guarantees both ex ante (i.e., in expectation) and ex post (i.e., in every allocation in the support of the probability distribution). \cite{AFSV24} initiated this line of work in (unconstrained) fair division, now referred to as ``best of both worlds'' (BoBW) guarantees. 

We denote the dot product of two vectors $\xx, \yy \in \bbR^d$ by $\xx \cdot \yy = \sum_{j \in [d]} x_j \cdot y_j$. A \emph{fractional} allocation is denoted by $\xx \in [0, 1]^{N \times M}$, where $x_{i, g}$ is the ``fraction'' of good $g$ allocated to agent $i$ and $\sum_{i \in N} x_{i,g} = 1$ for each $g \in M$. The utility of agent $i$ under a fractional allocation $A$ is $v_i(A_i) = \sum_{g \in M} x_{i,g} \cdot v_i(g)$. 

A randomized allocation $\xx = \sum_{\ell \in [L]} \lambda_\ell A^{\ell}$ is a probability distribution in which (integral) allocation $A^{\ell}$ is selected with probability $\lambda_\ell$ for each $\ell \in [L]$ and $\sum_{\ell \in [L]} \lambda_{\ell} = 1$. It induces a fractional allocation, also denoted $\xx$, in which $x_{i,g}$ is the marginal probability of agent $i$ being allocated good $g$. 

We say that a randomized allocation $\xx$ satisfies a desideratum $D$ \emph{ex ante} if its induced fractional allocation satisfies D,\footnote{We will only refer to envy-freeness (EF) and Pareto optimality (PO) of fractional allocations, defined exactly as they are for integral allocations. PO for a fractional allocation requires that not even a \emph{fractional} allocation Pareto dominate it.} and that $\xx$ satisfies a desideratum $D$ \emph{ex post} if every (integral) allocation in its support satisfies $D$. 

\subsection{EF+PO fractional allocations under linear constraints.} 

First, we define (a relevant special case of) the model of \cite{echenique2021constrained}. 

\begin{definition}[Linearly-Constrained Divisible Economy]\label{def:linear}
There is a set of agents $N = [n]$ and a set of divisible goods $M = [m]$, with each good $g$ having a \emph{supply} $q_g \in [1, n]$. For some $d \in \bbN$, there is a $d \times m$ matrix $A$ and a $d \times 1$ vector $b$, both with non-negative entries, such that, for a fractional allocation $\xx$, $\xx \in \allocs$ if and only if $A \xx_i \le b$ for each $i \in N$. 
\end{definition}

The supply is useful for modeling the case of goods with copies. Note that linear constraints can be used to capture goods with copies and/or balancedness (actually, even fractional versions of laminar matroids, but our BoBW results hold for only these two cases). For goods with copies, we set $q_g$ to be the number of copies of $g$ available, and demand $x_{i,g} \le 1, \forall g \in M$ for each $i \in N$, which is a set of linear constraints. To impose balancedness, we add $\sum_{g \in M} x_{i,g} \le \nicefrac{m}{n}$ for each $i \in N$; note that $\nicefrac{m}{n}$ does not need to be an integer and no lower bounds are required in this fractional case. 

\begin{definition}[Competitive Equilibrium From Equal Incomes (CEEI)]
\label{def:ceei}
For a fractional allocation $\xx$ and price vector $\pp \in \bbR^M_{\geq 0}$, $(\xx,\pp)$ is a competitive equilibrium from equal incomes (CEEI) if the following conditions hold:
\begin{enumerate}
    \item $x$ is feasible, i.e., $A x_i \le b$ for all $i \in N$.
    \item\label{cnd:cheapest-best-bundle} Each agent $i$ is allocated a bundle $\xx_i$ that maximizes her utility subject to a budget of 1, i.e.,
    \[
    \xx_i \in \arg\max_{\xx} \{v_i(\xx) \mid \pp \cdot \xx \leq 1\}.
    \]
    Further, $\xx_i$ should be the \emph{cheapest bundle} (minimizing $\pp \cdot \xx_i$) subject to achieving the maximum utility of $v_i(\xx_i)$.
    \item The market clears, i.e., $\sum_{i \in N} x_{i, g} \leq q_g$ for all $g$, and $\sum_{i \in N} x_{i, g} < q_g$ only if $p_g = 0$.
\end{enumerate}
\end{definition}

The following is a corollary of Theorem 1 of \cite{echenique2021constrained}.\footnote{Their stronger result holds for semi-strictly quasi-concave utilities, which subsumes strictly positive additive utilities.}

\begin{theorem}[\citeauthor{echenique2021constrained}~\citeyear{echenique2021constrained}]
\label{thm:ce-linear-constraints}
For a linearly-constrained divisible economy with strictly positive additive valuations,\footnote{That is, $v_i(g) > 0$ for all $i \in N$ and $g \in M$.} there always exists a fractional allocation $\xx$ and a price vector $\pp$ such that $(\xx,\pp)$ is CEEI. Further, in any such CEEI, $\xx$ is envy-free and Pareto optimal.
\end{theorem}

Our goal is to implement a fractional allocation $\xx$ that is part of a CEEI as a probability distribution over integral allocations that satisfy two relaxations of EF1 (satisfying them ex post): Prop1~\citep{CFS17} and \efoo~\citep{barman2019proximity}. 

\begin{definition}[Proportionality up to one good (Prop1)]
An integral allocation $A$ is \emph{proportional up to one good} (Prop1) if, for each agent $i \in N$, either $v_i(A_i) \geq \nicefrac{v_i(M)}{n}$ or there exists a good $g \notin A_i$ such that $v_i(A_i \cup \set{g}) \geq \nicefrac{v_i(M)}{n}$.
\end{definition}

\begin{definition}[Envy-freeness up to one good more-and-less (\efoo)]
An integral allocation $A$ is \emph{envy-free up to one good more-and-less} (\efoo) if, for every pair of agents $i, j \in N$ such that $A_j \ne \emptyset$ and $A_i \ne M$, $v_i(A_i \cup \set{g_i}) \geq v_i(A_j \setminus \set{g_j})$ for some $g_i \notin A_i$ and $g_j \in A_j$.
\end{definition}

We would also like our implementation of $\xx$ to be PO ex post. Luckily, this requirement is straightforward, as it is known that the implementation of any ex ante PO fractional allocation will be ex post PO.

\begin{proposition}[\citep{AFSV24} Proposition 1]
\label{prp:ex-ante-po-ex-post-po}
If a randomized allocation $\xx$ is ex ante PO, then $\xx$ is also ex post PO.
\end{proposition}

\subsection{Algorithm For Decomposing Constrained CEEIs}

\begin{algorithm}[t]
\caption{CE Lottery for goods-with-copies}\label{alg:bobw-copies}
\begin{algorithmic}[1]
\STATE $\xx \gets$ a competitive (fractional) allocation of \Cref{thm:ce-linear-constraints} due to \cite{echenique2021constrained}
\STATE For all $i \in N$ and $t \in [m]$, let $Q_{i, t} \gets \sum_{j \in [t]} x_{i, \pref_i(j)}$ be the total fraction of $i$'s top-$t$ goods allocated to her ($\pref_i(j)$ denotes agent $i$'s $j$-th most preferred good)
\STATE Define the bihierarchy set of constraints for an integral allocation $\yy$:
{\small
\begin{align*}
    \hierarchy_1 &= 
    \big\{\lfloor Q_{i, t} \rfloor \leq \sum\nolimits_{j \in [t]} y_{i, \pref_i(j)} \leq \lceil Q_{i, t} \rceil  \mid i \in N, t \in [m]
    \big\}, \\
    \hierarchy_2 &= 
    \big\{ \sum\nolimits_{i} y_{i, g} \leq q_{g} \mid g \in M 
    \big\} \cup \big\{ y_{i,g} \leq 1 \mid i \in N, g \in M\big\}.
\end{align*}
}
\STATE Decompose $\xx = \sum_{\ell = 1}^L \lambda_\ell \cdot \yy^\ell$ into $L$ integral allocations using the algorithm of \cite{BCKM13} (see their Appendix B) with input $\xx$ and $\hierarchy_1, \hierarchy_2$
\RETURN the randomized allocation $\sum_{\ell = 1}^L \lambda_\ell \cdot \yy^\ell$
\end{algorithmic}
\end{algorithm}

We decompose a fractional CEEI allocation using \Cref{alg:bobw-copies} to obtain the following.

\begin{restatable}{theorem}{bobw}\label{thm:bobw}
For goods with copies (resp., goods with copies with an additional balancedness constraint) and strictly positive valuations, there always exists a randomized allocation that is ex ante EF and PO, and ex post \efoo, Prop1, and PO (resp., Prop1 and PO).
\end{restatable}

Before proving \Cref{thm:bobw}, we will introduce several helpful tools, the first of which being the bihierarchy matrix decomposition framework of \citep{BCKM13}, which \Cref{alg:bobw-copies} (closely following the algorithm of \citep{AFSV24}) uses as a black-box.

The bihierarchy theorem states how a fractional allocation can be decomposed into a convex combinations of integral allocations in such a way that each integral allocation adhears to a set of constraints. In the language of \citep{BCKM13}, a constraint over an allocation $X$ takes the form of $\ubar{q}_S \le \sum_{(i,r) \in S}{X_{i,r}} \le \bar{q}_S$, where $S$ is a set containing agent-object pairs, and $\ubar{q}_S, \bar{q}_S \in \mathbb{N}$ are upper and lower quotas for the sum of the elements of $S$. An allocation $X$ satisfies a constraint $S$ if the statement $\ubar{q}_S \le \sum_{(i,r) \in S}{X_{i,r}} \le \bar{q}_S$ is true. A constraint structure $\mathcal{H}$ is a collection of constraints, and some allocation $X$ satisfies a constraint structure $\mathcal{H}$ if it satisfies all constraints in $\mathcal{H}$. The bihierarchy theorem is specifically for collections of constraints that form a bihierarchy structure.

\begin{definition}[Hierarchy and Bihierarchy \citep{BCKM13}]
    A constraint structure $\mathcal{H}$ is a hierarchy iff for any pair of sets $S, S' \in \mathcal{H}$, we have that $S \cap S' \in \set{\emptyset,S,S'}$. A constraint structure $\mathcal{H}$ is a bihierarchy iff there exists two hierarchies $\mathcal{H}_1, \mathcal{H}_2$ such that $\mathcal{H} = \mathcal{H}_1 \cup \mathcal{H}_2$ and $\mathcal{H}_1 \cap \mathcal{H}_2 = \emptyset$.
\end{definition}

\begin{theorem}[Bihierarchy Theorem \citep{BCKM13}]\label{thm:bihierarchy}
    For any fractional allocation $\xx$ and bihierarchy $\mathcal{H}$, if $\xx$ satisfies $\mathcal{H}$, then it can be represented as a randomized allocation $\xx = \sum_{\ell \in [L]} \lambda_\ell A^{\ell}$, where for each $\ell \in [L]$, $A^{\ell}$ satisfies $\mathcal{H}$. This randomized allocation can be found in strongly polynomial time.
\end{theorem}

We will also introduce the Utility++ guarantee of \citep{AFSV24}, which is what we use to ensure that the integral allocations produced by \Cref{alg:bobw-copies} have the desired fairness properties.

Through the bihierarchy presented in \Cref{alg:bobw-copies},  \cite{AFSV24} prove a utility guarantee that we restate below. Recall the definition of the two hierarchies $\hierarchy_1 \cup \hierarchy_2$,
\begin{align*}
    \hierarchy_1 &= 
    \left\{\lfloor Q_{i, t} \rfloor \leq \sum\nolimits_{j \in [t]} y_{i, \pref_i(j)} \leq \lceil Q_{i, t} \rceil  \mid i \in N, t \in [m]
    \right\}, \\
    \hierarchy_2 &= 
    \left\{ \sum\nolimits_{i} y_{i, g} \leq q_{g} \mid g \in M 
    \right\} \cup \left\{ y_{i,g} \leq 1 \mid i \in N, g \in M\right\},
\end{align*}
where $\pref_i(j)$ is the $j$th good in the preference ranking of $i$ (ties broken arbitrarily) and 
$Q_{i, t} = \sum_{j \in [t]} x_{i, \pref_i(j)}$ is the total fraction of goods allocated to $i$ in $x$ among her top $t$ goods. The first hierarchy for $i$ enforces that the number of goods $i$ receives among her first top $t$ goods in the integral allocations , i.e., $\sum_{j \in [t]} y_{i, \pref_i(j)}$, must be as close as possible to $Q_{i, t}$ (the total fraction of $i$'s top $t$ goods allocated to her in $x$). The second hierarchy $\hierarchy_2$ ensures that the integral allocations adhear to goods with copies constraints, where each agent recieves no more than a single copy of each good, and no more than $q_g$ total copies of any good $g \in M$ are allocated. By invoking \Cref{thm:bihierarchy} with the above bihierarchy set, we have the following utility guarantee.

\begin{lemma}[Utility Guarantee++ \citep{AFSV24}]
\label{lem:util-guarantee}
For a fractional allocation $\xx$ and the bihierarchy set of constraints in \Cref{alg:bobw-copies}, there exists a randomized allocation implementing $\xx = \sum_{\ell \in [L]} \lambda_\ell A^{\ell}$ such that for all integral allocations $A^{\ell}$ the following holds:
\begin{enumerate}
    \item If $v_i(A^{\ell}_i) < v_i(\xx_i)$, then there exists a good $g$ such that $x_{i,g} > 0$ and $v_i(A^{\ell}_i \cup \{g\}) > v_i(\xx_i)$.
    \item If $v_i(A^{\ell}_i) > v_i(\xx_i)$, then there exists a good $g$ such that $x_{i, g} < 1$ and $v_i(A^{\ell}_i \setminus \{g\}) < v_i(\xx_i)$.
\end{enumerate}
\end{lemma}

Finally, we introduce the concept of ``Bang-per-Buck'' in a CEEI. For an allocation $x$ and price vector $p$ such that $(x, p)$ is a CEEI, the \emph{bang-per-buck} ratio of good $g$ for agent $i$ is $\bb_i(g) = v_i(g) / p_g$. Bang-per-buck ratios have been found to be very useful when working with CEEIs of unconstrained instances with additive valuations, as each agent's CEEI bundle will contain only goods with the maximum bang-per-buck for that agent (see \citep{BKV18} for example). By the following lemma, we can show that even under goods with copies constraints, examining the bang-per-buck of each good will still allow us to give a similar structure to a CEEI solution.

\begin{lemma}
\label{lem:bb-relation}
In the case of goods with copies, given an allocation $\xx$ and price vector $\pp$ such that $(\xx, \pp)$ is CEEI, for any agent $i$, a good $g_1$ that is fully allocated to $i$ (i.e., $x_{i, g_1} = 1$), goods $g_2$ and $g_3$ that are partially allocated to $i$ (i.e., $x_{i, g_2}, x_{i, g_3} \in (0, 1)$), and $g_4$ that is not allocated to $i$ (i.e., $x_{i, g_4} = 0$), the following relation of bang-per-buck ratios holds:
\[
\bb_i(g_1) \ge \bb_i(g_2) = \bb_i(g_3) \ge \bb_i(g_4).
\]
\end{lemma}
\begin{proof}
Assume by contradiction that for a pair of goods $g, g' \in \{g_1, \ldots, g_4\}$ the claim does not hold, i.e., $\bb_i(g) < \bb_i(g')$ and $\xx_{i, g} > 0$. It holds that $g$ is either partially or fully allocated to $i$, while $g'$ is either partially or not at all allocated to $i$. 
Let $\delta_{g'} \in (0, 1)$ be some positive value such that $\delta_{g'} \le 1 - x_{i, g'}$ and $\delta_{g} \cdot p_{g'} / p_{g} \le x_{i, g}$. Take a bundle $\xx'_i$ such that $x'_{i, g'} = x_{i, g'} + \delta_{g'}$, $x'_{i, g} = x_{i, g} - \delta_{g'} \cdot p_{g'} / p_g$, and $x'_{i, g''} = x_{i, g''}$ for all $g'' \notin \{g, g'\}$. That is, agent $i$ ``spends'' a nonzero part of their budget on $g'$ instead of $g$. The inequalities ensure that such a transfer is feasible. Then, since $\bb_i(g') > \bb_i(g)$ and the rest of the spending of agent $i$ in $\xx'_i$ is the same as in $\xx_i$, we have $v_i(\xx'_i) > v_i(\xx_i)$, which contradicts condition \ref{cnd:cheapest-best-bundle} of \Cref{def:ceei}.
\end{proof}

Note that the proof above does not necessarily hold when a balancedness constraint is also enforced on top of the goods with copies constraint. For instance, the transfer made in the proof may lead to an unbalanced bundle $\xx'_i$ which is inadmissible.

With these definitions in mind, we are ready to prove \Cref{thm:bobw}.

\begin{proof}[Proof of \Cref{thm:bobw}]
Let $A^1, \ldots, A^{\ell}$ be the set of allocations returned by \Cref{alg:bobw-copies} based on the fractional allocation $\xx$ and price vector $\pp$ such that $(\xx, \pp)$ is CEEI. The ex-ante PO guarantee follows from \Cref{thm:ce-linear-constraints}, whose implementation is also ex-post PO by \Cref{prp:ex-ante-po-ex-post-po}.

First, note that the bihierarchy constraints $\hierarchy_1$ and $\hierarchy_2$ will ensure that the goods with copies (resp. balancedness) constraints will hold for each $A$ in the support of $\xx$. Any fractional goods with copies allocations will clearly adhere to the constraints in $\hierarchy_2$, and $\hierarchy_2$ will enforce that each integral allocation in the support will meet the goods with copies constraint. The constraints $\{ \sum\nolimits_{i} y_{i, g} \leq q_{g} \mid g \in M\}$ will ensure that no more than $q_g$ copies are allocated to each agent, and the constraints $\{ y_{i,g} \leq 1 \mid i \in N, g \in M\}$ will ensure that no single agent receives more than $1$ copy of any good.

The fact that the fractional allocation $\xx$ will adhere to the constraints in $\hierarchy_1$ is self-evident. Its primary purpose is to ensure our fairness guarantees hold ex post, however, note that $\hierarchy_1$ also guarantees that if $\xx$ is a balanced fractional allocation, then each of the integral allocations in its support will also be balanced. This is due to the fact that the constraints enforced by $\hierarchy_1$ for $t = m$ will be $\{\lfloor Q_{i, m} \rfloor \leq \sum\nolimits_{j \in [m]} y_{i, \pref_i(j)} \leq \lceil Q_{i, m} \rceil  \mid i \in N\}$. Here, when $\xx$ is balanced, we have that $Q_{i,m} = \sum\nolimits_{j \in M} x_{i,j} = m/n$ and $\sum\nolimits_{j \in [m]} y_{i, \pref_i(j)} = \sum\nolimits_{j \in [m]} y_{i,j}$.

The Prop1 guarantee follows from \Cref{lem:util-guarantee}, as in the proof of \cite{AFSV24}. Either \( v_i(A_i) \ge v_i(\xx_i) \ge v_i(M)/n \), where the last inequality holds because EF implies Prop, or by \Cref{lem:util-guarantee}, there exists a good \( g_i \notin A_i \) such that \( v_i(A_i \cup \set{g_i}) \ge v_i(\xx_i) \ge v_i(M)/n \).

Next, we prove the \efoo guarantee. Consider agents \( i \) and \( j \) such that \( v_i(A_i) < v_i(A_j) \). By \Cref{lem:util-guarantee}, there exists a good \( g_i \) with \( x_{i, g_i} < 1 \) such that 
\begin{equation}
\label{eq:util-i}
    v_i(A_i \cup \{g_i\}) > v_i(\xx_i),
\end{equation}
and a good \( g_j \) with \( x_{j, g_j} < 1 \) (i.e., partially allocated to \( j \)) such that 
\begin{equation}
\label{eq:util-j}
    v_j(A_j \setminus \{g_j\}) < v_j(\xx_j).
\end{equation}

We show that \( \pp(A_j \setminus \{g_j\}) < \pp(\xx_j) \). Let \( \xx^{\afull}_j \) be the bundle containing all the goods that are fully allocated to \( j \). We also have \( \xx^{\afull}_j \subseteq A_j \setminus \{g_j\} \), since the method of \cite{BCKM13} ensures that any good fully allocated to one agent remains allocated to the same agent in all integral allocations. Rewriting, we have \( v_j(A_j \setminus \{g_j\}) = v_j(\xx^{\afull}_j) + v_j(A_j \setminus (\xx^{\afull}_j \cup \{g_j\})) \) and \( v_j(\xx_j) = v_j(\xx^{\afull}_j) + v_j(\xx_j \setminus \xx^{\afull}_j) \). Combined with Eq.~\eqref{eq:util-j}, we obtain
\[
v_j(A_j \setminus (\xx^{\afull}_j \cup \{g_j\})) < v_j(\xx_j \setminus \xx^{\afull}_j).
\]

Both sides of the inequality consist of goods that are partially allocated to \( j \) in \( \xx \). By \Cref{lem:bb-relation}, all goods involved on both sides must have the same bang-per-buck ratio. Hence, by multiplying both sides by that ratio, we get
\[
\pp(A_j \setminus (\xx^{\afull}_j \cup \{g_j\})) < \pp(\xx_j \setminus \xx^{\afull}_j),
\]
and after adding \( \pp(\xx^{\afull}_j) \) to both sides, we have
\[
\pp(A_j \setminus \{g_j\}) < \pp(\xx_j) \le 1.
\]

Since \( A_j \setminus \{g_j\} \) fits within the budget of 1, by the best bundle condition of CE, \( v_j(A_j \setminus \{g_j\}) \le v_j(A_j) \). Along with Eq.~\eqref{eq:util-i}, we conclude that \( v_i(A_i \cup \{g_i\}) > v_i(A_j \setminus \{g_j\}) \), as desired.
\end{proof}

\subsection{Equivalence of \efoo for chores and goods-with-copies.}
It is clear to see that the guarantees we receive for goods with copies constraints directly imply the same guarantees in the unconstrained world, one simply has to set the number of copies for all goods to $1$. We also argue that this result implies the same guarantees for unconstrained instances with chores rather than goods.

For full context, the definition for \efoo for chores is as follows

\begin{definition}[\efoo for Chores]
An integral allocation $A$ is \emph{envy-free up to one good more-and-less} (\efoo) if, for every pair of agents $i, j \in N$ such that $A_i \ne \emptyset$ and $A_j \ne M$, $v_i(A_i \setminus \set{c_i}) \geq v_i(A_j \cup \set{c_j})$ for some $c_i \in A_i$ and $c_j \not\in A_j$.
\end{definition}

When \citep{GHLT24} introduced the goods with copies framework, they showed that when an allocation is \efwc in an instance where each good has $n-1$ copies, the allocation over the duel chores instance induced by the goods with copies will be EF1. \efwc is a stronger notion than EF1, and one may note that we are only claiming to find \efoo over goods with copies instances, rather than some analogous ``WC'' strengthening of it. In the below lemma, we will show that in fact, \efoo is a unique case of fairness property, where finding it in the goods with copies instance directly implies its existence in the duel chores instance.

\begin{lemma}
    The existence of an \efoo allocation in a goods with copies constrained instance with $n-1$ copies of each good implies the existence of an \efoo allocation in the duel chores instance.
\end{lemma}
\begin{proof}
    For contradiction, assume this is false. Let $A$ be an \efoo allocation over a goods with copies instance with $n-1$ copies of each good, and let $B$ be the allocation induced from $A$ in the duel chores instance (For every good $g$ with $n-1$ copies in the original instance, there is a chore $c_g$ in the duel instance with a value of $v_i(c_g) = -v_i(g)$ for all $i \in N$. The allocation $B$ is constructed by allocating each chore $c_g$ to the unique agent $j \in N$ such that $g \not\in A_j$). Assume that $B$ is not \efoo, so there exists a pair of agents $i,j \in N$, such that for all $c_i \in B_i, c_j \not \in B_j$, we have that $v_i(B_i \setminus \set{c_i}) < v_i(B_j \cup \set{c_j})$.

    Let the set $S$ represent a set containing a single copy of each good in the original goods with copies instance. From the way we constructed the duel chores instance, we must have that $v_i(B_i) = v_i(A_i) - v_i(S)$ for all $i \in N$. Also note that in $A$, we must have that $v_i(A_i \cup \set{g_i}) \geq v_i(A_j \setminus \set{g_j})$ for some goods $g_i \not\in A_i, g_j \in A_j$. 

    This fact implies that we can take the two corresponding chores in the duel instance $c_{g_i} \in A_i, c_{g_j} \not\in A_j$, such that $v_i(c_{g_i}) = -v_i(g_i)$ and $v_i(c_{g_j}) = -v_i(g_j)$, and from this, we can rewrite the $i,j$ \efoo guarantee from the goods with copies instance as the following:
    \[v_i(B_i) + v_i(S) - v_i(c_{g_i}) \geq v_i(B_j) + v_i(S) + v_i(c_{g_j})
    \]
    removing the $v_i(S)$ from both sides get us
    \[
    v_i(B_i) - v_i(c_{g_i}) \geq v_i(B_j) + v_i(c_{g_j})
    \]
    which means that $v_i(B_i \setminus \set{c_{g_i}})  \geq v_i(B_j \cup \set{c_{g_j}})$, giving a contradiction. 
\end{proof}

\subsection{Extensions to Broader Constraints}

Note the missing guarantee of \efoo when balancedness is imposed in \Cref{thm:bobw}. Our approach closely follows the decomposition of a fractional CEEI allocation in the \emph{unconstrained} case due to \cite{AFSV24}; however, there are key differences that we highlight. While our Prop1 and PO guarantees follow directly from lemmas they prove (\Cref{prp:ex-ante-po-ex-post-po} and \Cref{lem:util-guarantee}), our \efoo guarantee does not. They use a property of fractional MNW allocations, which are CEEI in the unconstrained case. We need to establish an appropriately relaxed version of this property (\Cref{lem:bb-relation}) for CEEI allocations in our constrained case, which may be of independent interest. Finally, our algorithm is similar to theirs, with a slightly different $\hierarchy_2$ hierarchy in the bihierarchy rounding of \cite{BCKM13} due to working with goods with copies. As stated in the proof of \Cref{thm:bobw}, $\hierarchy_2$ will enforce that each integral allocation in the support of the CEEI will meet the goods with copies constraints, while $\hierarchy_1$ enforces that if the CEEI is also balanced, then each integral allocation in the support will be balanced as well.

Goods with copies and balancedness constrains both have simple structures that allow their enforcement to be integrated easily into the existing structure of the hierarchies from the Utility++ guarentee. Unfortunately, this does not seem to be the case for more general constraints. Take partition matroids for example. The CEEI theorem of \cite{echenique2021constrained} is general enough to allow us to find a fractional CEEI that meets a given set of partition matroid constraints, however, there is no way that we know of to break down this fractional allocation into a convex combination of integral allocations which also meet the given constraints.

The easiest way to attempt to do this would be to add additional constraints to $\mathcal{H}$ to ensure that the partition matroid constraints are met. For a given partition matroid $\calC$, such a constraint set would look like:

\begin{align*}
    \hierarchy_3 &= \left
    \{0 \leq \sum\nolimits_{j \in C} y_{i,j} \leq h_C  \mid i \in N, C \in \mathcal{C}
    \right\}, \\
\end{align*}

But unlike balancedness or goods with copies, this new constraint set does not fit into one of the existing hierarchies $\mathcal{H}_1$ and $\mathcal{H}_2$. Enforcing them in addition to the other two hierarchies (which are both necessary for the Utility Guarantee++ lemma to work) would create a tri-hierarchy constraint, something which is not covered by the results of \cite{BCKM13}.

\section{Discussion}\label{sec:discussion}
Our work is only the first step towards a systematic study of fair and efficient allocations under feasibility constraints. In addition to resolving the major open questions listed in \Cref{sec:implications}, the field pursue several exciting research directions. 

\paragraph{Computation.} We did not discuss how to compute MNW allocations subject to any of the various feasibility constraints we study. It is likely a hard problem in each case, as it is for the unconstrained case~\citep{Lee15}. For the unconstrained case, the literature offers polynomial-time constant-approximations of the Nash welfare objective (the current best being a 1.45-approximation due to \cite{BKV18}), often together with (approximate) fairness guarantees~\citep{mcglaughlin2020improving,chaudhury2022fair}. Such a result under feasibility constraints would be exciting.  

\paragraph{Pushing for stronger BoBW guarantees.} In contrast to our deterministic guarantees from \Cref{sec:nash}, our BoBW guarantees from \Cref{sec:bobw} hold only for the special cases of goods with copies and/or balancedness. Can such guarantees be established even for partition matroid constraints? We also understand the limits of fractional allocations a lot less in the constrained setting. To that end, we make an observation for linearly-constrained divisible economies (\Cref{def:linear}). While CEEI allocations exist and achieve EF and PO (\Cref{thm:ce-linear-constraints}), fractional MNW allocations, which may not be CEEI, are still \half-EF and PO. This is because \cite{EFS24} establish that fractional MNW allocations among a convex constraint set have an individual harm ratio of $1$, and this immediately implies \half-EF, generalizing a recent result by \cite{troebst2024cardinal} that a fractional MNW \emph{matching} (i.e., when $m=n$ and balancedness is imposed) is \half-EF and PO. It is also worth noting CEEI allocations can achieve the appealing guarantees of EF and PO, there are desirable fairness and efficiency guarantees that cannot be achieved simultaneously, even by fractional allocations under a partition matroid constraint~\citep{KawaseSY23}.  

\paragraph{Extending \efwc beyond goods with copies.} \efwc is an appealing strengthening of EF1, but defined narrowly for the case of goods with copies. Can we define viable strengthenings of EF1 for more general settings (e.g., partition matroid constraints, or even arbitrary matroid constraints), which reduce to \efwc for the special case of goods with copies?

\section*{Acknowledgements}
This research was partially supported by an NSERC Discovery grant.

\bibliography{abb,ultimate,nisarg}

\clearpage
\appendix
\section*{\centering Appendix}

\section{Omitted Details from \Cref{sec:nash}}\label{app:nash}

\subsection{Inefficiency of Round Robin}\label{app:round-robin-inefficient}

Through the following example, we will show an instance where no round-robin allocation is PO.

\begin{example}\label{ex:roundrobin-PO-NO}
    Consider an instance with two agents and $8$ goods, where the agents have the following preferences:

    \begin{table}[ht]
        \centering
        \begin{tabular}{c | c c c c c c c c}
            \toprule
                   & $g_1$ & $g_2$ & $g_3$ & $g_4$ & $g_5$ & $g_6$ & $g_7$ & $g_8$ \\
            \midrule
             Agent 1 & 10    & 9     & 5     & 4     & 3     & 2     & 1     & 0     \\
             Agent 2 & 10    & 9     & 8     & 7     & 6     & 5     & 1     & 0     \\
             \bottomrule
        \end{tabular}
    \end{table}

    Observe the best possible bundle that each agent can receive under a round-robin ordering. Since both of the agents have the same strict ordering of the goods ($v_i(g_1) > v_i(g_2) > \dots > v_i(g_8)$ for both $i \in [2]$), the best possible round-robin bundle for each agent will be the bundle $B = \set{g_1,g_3,g_5,g_7}$. For agent $1$, we have $v_1(B) = 19$, and for agent $2$, we have $v_2(B) = 25$.

    However, consider the allocation $A = (\set{g_1,g_2,g_7,g_8},\set{g_3,g_4,g_5,g_6})$. Agents preferences in $A$ will be $v_1(A_1) = 20 > v_1(B)$ and $v_2(A_2) = 26 > v_2(B)$, showing that this allocation will Pareto Dominate any round-robin allocation.
\end{example}

\subsection{Lower Bound for MNW Fairness Guarentees with Good with Copies Constraints}\label{app:lowerboundcopies}

\begin{restatable}{theorem}{copieslowerbound}\label{thm:copieslowerbound}
    For any $\epsilon > 0$, there exist instances with the goods-with-copies constraint in which no MNW (or complete MNW) allocation is $(\half + \epsilon)$-EF1.
\end{restatable}
\begin{proof}
    Let $k \in \bbN$ and fix $\delta < \nicefrac{1}{k}$. Consider the following instance with $2$ agents and $2k$ goods. The goods are partitioned into $S$ and $T$, each consisting of $k$ goods. There are $2$ copies of each of the $k$ goods in $S$, and for each good $g_s \in S$, we have $v_1(g_s) = \delta$ and $v_2(g_s) = 1$. There is a single copy of each good $g_t \in T$, with agent valuations being $v_1(g_t) = v_2(g_t) = 1$.

    Consider any allocation $A$. First, notice that since there are two copies of every good in $S$, the goods-with-copies constraint enforces that each agent receive one copy of each good in $S$. Suppose $x = |A_1 \cap T|$. Then, we have $v_1(A_1) = k\delta + x$ and $v_2(A_2) \le k + (k-x) = 2k - x$. Thus, $\NW(A) \le (k\delta + x) \cdot (2k - x)$. The right hand side is uniquely maximized at $x = k - \frac{k\delta}{2}$, and, given the choice of $\delta$, it is easy to see that the integral value of $x$ that uniquely maximizes it is $x = k$: $x = k$ yields a value of $k^2 \cdot (1+\delta)$, and $(k\delta + x) \cdot (2k - x) \ge k^2 (1+\delta)$ implies $x \in [k - k\delta, k]$, with $k$ being the only integer in that range due to our choice of $\delta$. Consider the allocation $A^* = (S \cup T, S)$, corresponding to $x=k$, in which agent $1$ receives one copy of each good in $S$ along with all of $T$ while agent $2$ receives one copy of each good in $S$. Since this achieves the maximum right hand side value, it must be the unique MNW (and complete MNW) allocation in this instance. 
    
    In this allocation, agent $1$ is not envious, but, for agent $2$, we have $v_2(A^*_2) = k$ and $v_2(A^*_1 \setminus \set{g}) = 2k-1$ for every $g \in A^*_1$. Hence, this allocation is $\frac{k}{2k-1}$-EF1, i.e., $(\frac{1}{2} + \frac{1}{4k-2})$-EF1. It is clear that for any $\epsilon > 0$, we can choose $k > \frac{1}{2} + \frac{1}{4\epsilon}$ to ensure that the allocation is not $(\half + \epsilon)$-EF1.
\end{proof}

\section{Matroid Preprocessing Step From \cite{dror2023fair}}\label{app:preprocessing}

The process of \cite{dror2023fair} involves performing a number of ``Free Extensions'' on the instance. For some matroid $\mathcal{M} = (M,\mathcal{I})$ with rank $r$, performing a free extension on $\mathcal{M}$ means creating a new matroid $\mathcal{M'} = (M',\mathcal{I}')$ by adding a new good $x$, such that $M' = M \cup \set{x}$ and $\mathcal{I} = \set{I \cup \set{x} | I \in \mathcal{I}, |I| < r}$, i.e., it adds $x$ to all independent sets that have size less than the rank of the matroid. Note that a free-extension will never increase the rank of the matroid.

The preprocessing step will keep preforming free extensions on $\mathcal{M}$, until we end up with a new matroid $\mathcal{M'} = (M',\mathcal{I}')$ where $|M'| = nr$. Each free extension will involve us adding a ``dummy good'' to $M$ (a good that each agent has a valuation of $0$ for). One can immediately see that since the rank of the of $\mathcal{M'}$ is $r$, any complete and feasible allocation of $M'$ must give exactly $r$ goods to each agent.

\cite{dror2023fair} prove that the matroid $\mathcal{M'}$ constructed in the preprocessing step provides the following guarantees:
\begin{observation}[Observation 2 of \cite{dror2023fair}]
    The matroid $\mathcal{M'}$ constructed in the preprocessing step provides the following guarantees:
    \begin{itemize}
        \item All bases of $\mathcal{M}$ are bases of $\mathcal{M}'$
        \item For every complete allocation $A' \in \mathcal{F}'$, $|A'_i| = r$ for all $i \in N$. Therefore, every $A'_i$ is a basis of $\mathcal{M}'$.
        \item For every complete allocation $A \in \mathcal{F}$, there is a complete allocation $A' \in \mathcal{F}'$, where for all $i \in N$, $A'_i$ contains $A_i$ plus zero or more dummy goods, Therefore, $v_j(A_i) = v_j(A'_i)$ for all $j \in N$.
        \item For every complete allocation $A' \in \mathcal{F}'$, there is a complete allocation $A \in \mathcal{F}$, where for all $i \in N$, $A_i$ equals $A'_i$ with all dummy goods removed. Therefore, $v_j(A'_i) = v_j(A_i)$ for all $j \in N$.
        \item $\mathcal{M}'$ is base-orderable if and only if $\mathcal{M}$ is base-orderable. 
    \end{itemize}
\end{observation}

Let $\mathcal{F}$ be the feasibility set of an instance constrained by a base-orderable matroid $\mathcal{M}$, and let $\mathcal{F}'$ be the feasibility set of the instance constrained by $\mathcal{M}'$, the matroid constructed by performing the preprocessing step on $\mathcal{M}$.

We can say that an allocation $A' \in \mathcal{F}'$ will be a complete MNW allocation for $\mathcal{F}'$ if and only if the corresponding allocation $A \in \mathcal{F}$, constructed by removing all the dummy goods from $A'$, is a complete MNW allocation for $\mathcal{F}$.

First, in the case where an MNW allocation $A' \in \mathcal{F}'$ has $NW(A') > 0$, this follows from the fact that all agents' utilities remain the same after dummy goods are removed. Thus, for the allocation $A \in \mathcal{F}$ formed by removing the dummy goods from $A'$, we must have $NW(A) = NW(A')$. If there were another allocation $A^* \in \mathcal{F}$ such that $NW(A^*) > NW(A)$, then we could add dummy goods to each bundle of $A^*$ until $|A^*_i| = r$ for all $i \in N$. This new allocation must be in $\mathcal{F}'$, and must have higher Nash Welfare than $A'$, giving a contradiction. Conversely, if $A \in \mathcal{F}$ maximizes Nash Welfare for $\mathcal{F}$, with $NW(A) > 0$, but its corresponding allocation $A' \in \mathcal{F}'$ does not, then the same contradiction can be arrived at by finding the MNW allocation for $\mathcal{F}'$ and removing the dummy goods.

Next, in the case where an MNW allocation $A' \in \mathcal{F}'$ has $NW(A') = 0$, it must be the case that $A'$ maximizes the number of agents that receive positive utility, and with respect to this, maximizes the product of positive utilities. Notice that after removing the dummy goods from $A'$, we get an allocation $A$ where the exact same agents get positive utility, and the product of those agents' utilities are the same. If there is some $A^* \in \mathcal{F}$ where $NW(A^*) > 0$, then that gives a contradiction, as adding dummy goods to $A^*$ would create an allocation in $\mathcal{F}'$ with positive Nash Welfare. If there is an allocation in $A^*$ where more agents receive positive utility, or the product of positive utilities is higher, then we can add dummy goods to create an allocation from $\mathcal{F}'$ with the same guarantees, giving another contradiction. Conversely, if $A \in \mathcal{F}$ is the MNW allocation for $\mathcal{F}$ and $NW(A) = 0$, but the corresponding allocation $A' \in \mathcal{F}'$ is not MNW, then the same contradictions can be arrived at by finding the MNW allocation for $\mathcal{F}'$ and removing the dummy goods.
\end{document}